\documentclass[12pt]{amsart}

\usepackage{amssymb,amsmath,amsfonts,amstext,latexsym}
\usepackage[all]{xy}
\usepackage{color}
\usepackage{pdfsync}
\usepackage{amssymb}
\usepackage{paralist} 
\usepackage{stmaryrd}
 \usepackage{graphicx}
 \usepackage{tikz, xcolor}
\usetikzlibrary {shapes}
 \usepackage{array}
 \usepackage{fge}
 \usepackage{multicol}
 \usepackage{caption,pgf} 
 \usepackage{rotating}
 \usepackage{adjustbox}
\usepackage{blindtext}

\usepackage[super]{cite} 

\usepackage{setspace}
\newtheorem{thm}{Theorem}
\newtheorem{lemma}[thm]{Lemma}

\newtheorem{theorem}[thm]{Theorem}
\newtheorem{corollary}[thm]{Corollary}

\newtheorem{definition}[thm]{Definition}

\def\X{{\mathcal X}}

\def\phiha{\hat{\phi}}

\def\om{\omega}

\def\bfx{\boldsymbol{x}}
\def\bfmu{\boldsymbol{\mu}}
\def\bfy{\boldsymbol{y}}
\def\bfvartheta{\boldsymbol{\vartheta}}

\def\bfM{\boldsymbol{M}}

\def\RR{{\mathbb R}}

\def\NN{{\mathbb N}}
\def\ZZ{{\mathbb Z}}

\def\implies{\Rightarrow}
\def\Mut{\mathtt{Mut}}
\def\Sel{\mathtt{Sel}}
\def\Proje{\mathtt{Proj}}

\def\Truncation{\texttt{Trunc}}
\def\kati{\tilde{\kappa}}
\def\Vati{\tilde{V}}
\def\Wha{\hat{W}}
\def\Co{{\textrm{Co}}}

\def\bfphi{{\boldsymbol{\phi}}}
\def\bfpsi{{\boldsymbol{\psi}}}
\def\Mut{\mathtt{Mut}}
\def\Sel{\mathtt{Sel}}
\def\Rec{\mathtt{Rec}}
\def\MF{M}

\def\Var{\mathop{\mathrm{Var}}}

\def\sextt{{\texttt{sex}}}
\def\asextt{{\texttt{asex}}}
\def\GM{\mathop{\mathrm{GM}}}
\def\GMsup{\mathop{\overline{\GM}}}
\def\GMinf{\mathop{\underline{\GM}}}

\def\and{\mbox{and}}

\title[The evolutionary benefits of sexual reproduction]{A mathematical analysis of the evolutionary benefits of sexual reproduction}

 \author{Andrew E.M. Lewis-Pye$^{1}$ \& Antonio Montalb\'{a}n$^2$}
 
 \thanks{This paper has no principal author. The ordering is alphabetical.  Both authors contributed equally to the construction of proofs and simulations.}

\begin{document}

\maketitle



\begin{abstract}
The question as to why most higher organisms reproduce sexually has remained open despite extensive research, and has been called ``the queen of problems in evolutionary biology''.   Theories dating back to Weismann have suggested that the key must lie in the creation of increased variability in offspring, causing enhanced response to selection. Rigorously quantifying the effects of assorted mechanisms which might lead to such increased variability, and establishing that these beneficial effects outweigh the immediate costs of sexual reproduction has, however, proved problematic. Here we introduce an approach which does not focus on particular mechanisms influencing factors such as the fixation of beneficial mutants or  the ability of populations to deal with deleterious mutations, but rather tracks the entire distribution of a population of genotypes as it moves across vast fitness landscapes. In this setting simulations now show sex robustly outperforming asex across a broad spectrum of finite or infinite population models. Concentrating on the additive infinite populations model, we are able to give a rigorous mathematical proof establishing that sexual reproduction acts as a more efficient optimiser of mean fitness, thereby solving the problem for this model. Some of the  key features of this analysis carry through to the finite populations case.      
\end{abstract}

Sexual propagation must certainly confer immense benefits on those populations undergoing it, given that sex involves substantial costs such as the breaking down of favourable gene combinations established by past selection. Hypotheses as to the form  these advantages take fall naturally into two groups \cite{Kond,May1,Fel}. On the one hand a function of sexual reproduction and meiotic recombination may be in providing immediate and physiological benefits, such as allowing repair of double strand DNA damage \cite{Bern,Mich}.  Such mechanisms alone, however, are unlikely to account for the continued prevalence of sexual reproduction \cite{BC,Kond,May2}, and so, on the other hand, decades of research have seen evolutionary biologists looking to develop explicit theoretical models which  explain the advantages of sex in terms of the interaction between variation and selection. Many of these models \cite{HR,OB,BO} focus on ideas originally due to Morgan \cite{Morg}, Fisher \cite{Fish} and Muller \cite{Mull}, which stress the ability of recombination to place beneficial mutations together on the same chromosome. In a similar vein one may also consider the build up of deleterious mutations \cite{Mull2,Fel}. Sex has been shown to be  favoured in certain models which allow for fluctuating environments \cite{SKB,PL,GO} (e.g.\ cycling between negative and positive epistasis), so long as fluctuations are sufficiently rapid.   When the fitness of genotypes depends upon their geographical location, sex may also evolve under the right circumstances to break down locally detrimental genetic associations created by migration \cite{PZF,LO}. In general though, while great strides have been made in our understanding of some of the principal population genetic mechanisms yielding an advantage to sex\cite{Otten}, the present analysis tends to require rather specific conditions for these mechanisms to be of relevance, and does not provide a setting which allows one to establish that the benefits conferred outweigh the acute costs.   

Our aim here is to establish  a setting in which sex is seen to robustly outperform asex, even in the absence of epistasis and across a broad spectrum of models. We consider a setting in which a population of genotypes evolves over time in a context where there are no apriori limits on the number of alleles or their fitnesses. Figure 1 shows a small cross-section of the results of simulations for models with finite or infinite haploid populations and where fitness contributions from individual genes may be combined additively or multiplicatively (further examples are given in Figures 6-10 \S 4 Extended Data). In all cases the sexual population is seen to quickly achieve and maintain higher mean fitness, indicating that selection will favour genes coding for sexual rather than asexual reproduction -- the one exception being the boundary case of the multiplicative infinite populations model, in which the sexual and asexual processes remain identical in the absence of any initial linkage disequilibrium. 

We then concentrate our mathematical analysis on the infinite populations additive model, since dealing with this case allows us to avoid some of the complexities inherent in the finite population models while illustrating basic principles which carry through to the finite population additive model. We are able to give a mathematical proof that, during the process of asexual propagation, a negative linkage disequilibrium will be created and maintained, meaning that an occurrence of recombination at any stage of the process will cause an immediate increase in fitness variance and a corresponding increase in the rate of growth in mean fitness. In contexts where there is a large but finite bound on allele fitnesses, we prove that the sexual population will always be that which survives when sexual and asexual populations compete for resources.

\begin{figure}[!ht]
  \begin{adjustbox}{addcode={\begin{minipage}{\width}}
  {\caption{The dominance of sexual populations. Each plot corresponds to one simulation and shows the evolution of mean fitness (top) and fitness variance (bottom) over time for sexual and asexual populations beginning in identical states at linkage equilibrium. Finite models are described in (SI). Each simulation is specified by a tuple $(P,\ell,p,q,I)$:  $P\in \{ \{\infty \} \cup \mathbb{N} \}$ is population size;  $\ell$ is number of loci; $p$ is the probability of mutation; $q$ is the probability a given mutation is beneficial (=0.1 in all simulations displayed here); $I$ is initial fitness of alleles. For the infinite multiplicative model (e) the sexual and asexual populations stay identical. For the finite multiplicative model (f) the fitness values have been divided by $10^\ell$.
 }\end{minipage}},rotate=90,center}
  \includegraphics[scale=.72]{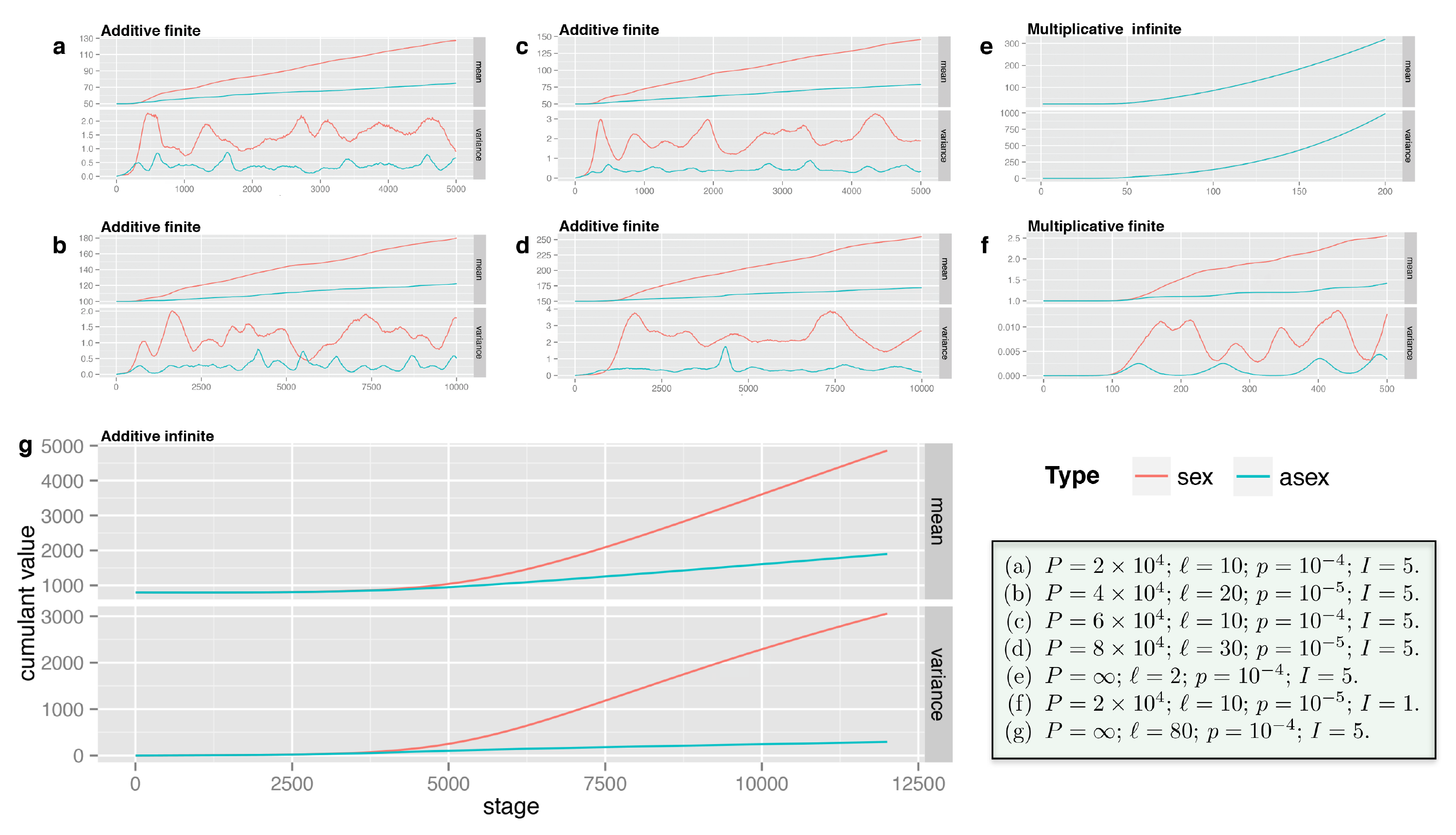}
  \end{adjustbox}
\end{figure}


\section{The model} 
We consider haploid populations with non-overlapping generations. In the absence of epistasis between alleles at a single locus, this analysis could easily be extended to consider diploid populations. We describe here the additive infinite population variants of the model (other variants are described in (SI) \S\ref{variants}). A slightly unusual feature of the model is that we do not assume  alleles come from a pre-existent pool, but consider a (form of random walk mutation) model in which alleles are created by mutation as time passes, possibly without any bound on attainable fitness. 
We shall assume that genes fitnesses take integer values, but one could also consider real valued fitnesses without substantial changes in behaviour.  Most other features of the model, which we now describe in more detail, are essentially standard in the literature.  

\

Each instance of the model is determined by three principal parameters: $\ell$, $D$ and $\mu$. 
{First},  $\ell\in \NN\  (>1)$ specifies the number of loci. 
With each individual specified by $\ell$ genes, in the absence of epistasis we need only be concerned with the fitness values corresponding to those genes, and so each individual can be identified with a tuple $\bfx=(x_1,...,x_\ell)\in \ZZ^\ell$. 
The {\em fitness} of $\bfx$ is $F(\bfx)=\sum_{i=1}^\ell x_i$. (For the multiplicative model, one would define $F(\bfx)=\prod_{i=1}^\ell x_i$ instead.) 
{Second}, the {\em domain} $D\subset \ZZ^\ell$ determines which individuals are allowed to exist.
We will use three types of domains in this paper: The {\em $\NN$-model} uses as domain $D=\NN^\ell$, where $\NN=\{1,2,3,....\}$; the {\em $\ZZ$-model} uses $D=\{\bfx\in \ZZ^\ell: F(\bfx)>0\}$; and the {\em bounded-model}  uses $D=\{1,...,N\}^\ell$ for some upper bound $N\in\NN$ on gene fitnesses.
In practice there is almost no difference between the $\NN$- and $\ZZ$-models, but there are situations when it is simpler to consider one or the other. 
{Third}, $\mu\colon\ZZ\to \RR^{\geq 0}$, the {\em mutation probability function}, determines how mutation affects the fitness of genes: $\mu(k)$ is the probability that the fitness of a gene will increase by $k$. 
For the sake of simplicity we assume this distribution to be identical for all loci. 
While there is no clear canonical choice for $\mu$,  the behaviour of the model is robust to changes in this parameter so long as negative mutations are more likely than positive ones,  both being possible.   
This is because any such choice of $\mu$ will approximate a Gaussian distribution over multiple generations.
The simplest mutation distributions one may consider  are those taking non-zero values only on $\{-1,0,1\}$. 
Unless stated otherwise, it should be assumed that from now on mutations are of this form and that $\mu(0)>\mu(-1)>\mu(1)$ (giving a form of stepwise-mutation model \cite{OK}).

By a {\em population} we mean  a probability distribution $\bfphi\colon \ZZ^\ell\to \RR^{\geq 0}$, where $\bfphi(\bfx)$ is the proportion of individuals that have  `genotype'  $\bfx\in \ZZ^\ell$. For a population $\bfphi$, we shall also use $X=(X_1,...,X_{\ell})$, where the $X_i$'s take values in $\ZZ$, to denote a random variable that picks an individual with gene fitness values $X_1$,...,$X_\ell$ according to the distribution given by $\bfphi$. We let $M(\bfphi)$ denote the {\em mean fitness} of the population $\bfphi$, namely $E(F(X))$. It should be assumed throughout that all populations considered have finite means, variances, and that all cumulants are finite (as is the case, for example, for distributions $\bfphi$ with finite support, i.e. with finitely many $\bfx\in \ZZ^\ell$ such that $\bfphi(\bfx)\neq 0$).

For a sexual population, the next generation is obtained by application of three operations:  selection, mutation and recombination. 
We refer to the consecutive application of these operations over multiple generations as the {\em sex process}.
For the {\em asex process}, the operations applied are selection and mutation, and the recombination phase is omitted.
With a much less significant effect, at the end of each generation we will also apply a truncation operation that erases  individuals  falling outside the domain.

\noindent \textbf{Selection}. 
The probability of survival for an individual is proportional to its fitness value.
If $\bfphi$ is the population prior to selection then the resulting population, $\Sel(\bfphi)$, is given by:
\[
\Sel(\phi)(\bfx)=\frac{F(\bfx)}{M(\bfphi)}\bfphi(\bfx), 
\quad \mbox{ for } \bfx\in \mathbb{\ZZ}^\ell.
\]  
The factor $1/M(\bfphi)$  normalises the probability distribution.

\noindent \textbf{Mutation}. 
Let $C_i$ be i.i.d.\ random variables taking values in $\ZZ$ with distribution $\mu$.
If we apply mutation to a random variable $X=(X_1,...,X_\ell)$ we get $(X_1+C_1,...,X_\ell+C_\ell)$.
Equivalently, if $\bfphi$ is the population prior to mutation then, for $ \bfx\in \mathbb{\ZZ}^\ell$: 
\[ 
\Mut(\bfphi)(\bfx)=\sum_{\boldsymbol{y}\in D} \bfphi(\boldsymbol{y})\cdot \bfmu(\boldsymbol{y}-\bfx), 
\] 
where $\bfmu$ is the extension of $\mu$ to a function on $\mathbb{Z}^\ell$ according to the assumption that mutations act independently on distinct loci (i.e., $\bfmu(a_1,...,a_\ell)=\prod_{i=1}^\ell\mu(a_i)$).

\noindent \textbf{Recombination}. For the sake of simplicity we assume that the $\ell$ loci are unlinked, so that they either correspond to loci on distinct chromosomes (one may consider that we are choosing a `representative' from each chromosome), or else lie at sufficient distances when they share a chromosome. In general the effect of recombination is to leave the distributions at individual loci unchanged, while bringing the population towards linkage equilibrium. We make the simplifying assumption (for the infinite models) that the effect of a single application of recombination is to bring the population immediately to linkage equilibrium. (A population is at {\em  linkage equilibrium} if the random variables $X_i$ are independent.)
If $\phi_i(x)\colon \ZZ\to \RR^{\geq 0}$ is the distribution at locus $i$, (i.e.\ $\phi_i(x)=\sum_{\bfy\in D, y_i=x}\bfphi(\bfy)$) then the resulting population is given by: 
\[
\Rec(\bfphi)(\bfx)=\prod_{i=1}^\ell \phi_i(x_i),
\quad \mbox{ for }\bfx=(x_1,...,x_\ell)\in\ZZ^\ell. 
\]
Recombination as we consider it here is thus equivalent to multiple applications of recombination in its standard form. 

\def\Tru{\mathtt{Tru}}
   Mutation and recombination may create individuals that fall outside the domain $D$. At the end of each generation, we therefore perform {\em truncation}  to remove those outlying individuals. $\Tru(\bfphi)(\bfx)$ is defined to be $\bfphi(\bfx)/s$ if $\bfx\in D$, and $0$ otherwise, where $s$ is the normalising factor $s=\sum_{\bfx\in D}\bfphi(\bfx)$.
We will see (Tables 1-3, \S 4 Extended Data) that the proportion of the population moving outside the bounds of $D$ in each generation is negligible, and that truncation along the lower bounds will have an insignificant effect on the whole process.


\section{Analysing the model}\label{se: Basic Analysis}

The objective now is to show that mean fitness increases more rapidly for sexual populations (reduction to selection at the gene level can then be achieved in a standard fashion, by consideration of the effect of selection on genes which code for sexual rather than asexual reproduction). Proofs of all claims in this section appear in (SI).

Each generation sees two forces acting on the mean fitness $M=M(\bfphi)$.
On the one hand, mutation causes a fixed decrease in $M$ by an amount that depends only on $\mu$. (Recall that deleterious mutations are more likely than beneficial ones.) Selection, on the other hand, can be shown to increase mean fitness by $V/M$ (a form of Fisher's `fundamental theorem' \cite{Fish}), where $V=V(\bfphi)=\Var(F(X))$ is the {\em variance of the fitness} of $\bfphi$. Recombination does not affect $M$ directly. Thus, for fixed $\mu$, the increase in mean fitness at each generation is determined by the variance.  
The difference between the sex and asex processes will be seen to stem from the effect of recombination on variance, which then results in an increase to the change in mean fitness for the sex process during the selection phase.

The effect of mutation on the variance is a fixed increase at each generation (again entirely determined by $\mu$). 
The effect of selection on variance is given by:
\[
V(\Sel(\bfphi))-V(\bfphi) = \frac{\kappa_3}{M}- \left(\frac{V}{M}\right)^2,
\]
where $\kappa_3$ is the third cumulant of $F(X)$.
Our first  theorem shows that for the sex process, the effect of recombination on variance is positive, giving an advantage of sex over asex.

\begin{theorem}\label{thm: effect rec}
If $\bfphi^{\ast}=\Sel(\bfphi)$ was obtained by an application of selection to a population $\bfphi$ at linkage equilibrium, then 
the effect of recombination on fitness variance is given by:
\[
V(\Rec(\bfphi^{\ast}))-V(\bfphi^{\ast}) = \frac{\sum_{i\neq j} V_iV_j}{M^2},
\]
where $V_i=\Var(\phi_i)$ and $M=M(\bfphi)$. This effect is therefore non-negative. 
\end{theorem}

This theorem applies to the sex process because a previous application of recombination would bring the population $\bfphi$ to linkage equilibrium.
Linkage equilibrium is then preserved by mutation.

Our second theorem shows that recombination has a positive effect on variance in a much more general situation, as for instance, during an asex process where we suddenly apply recombination.
It establishes that for a population initially at linkage equilibrium, \emph{any} subsequent applications of  recombination during later generations always give an increase in variance and so a corresponding increase in the rate of change of mean fitness.

\begin{theorem}\label{thm: rec positive}
For the $\ZZ$-model, starting with a population at linkage equilibrium, suppose we iterate the operations of mutation, selection and recombination in any order (possibly applying only mutation and selection over multiple generations, and of course applying truncations when relevant). 
Then any non-trivial application of recombination has a positive effect on variance. 
\end{theorem}
  
\noindent By a {\em trivial} application of recombination we mean one acting on a population which is already at linkage equilibrium, and so which has no effect at all. This is the case, for instance, if one applies recombination twice in a row: the second application is trivial.
The theorem is stated only for the $\ZZ$-model because truncation creates technical difficulties when producing a proof for the other models. With the effect of truncation being so small, however, the claim of the theorem is, in fact, verified in all simulations we have run for any of the models.

To explain what is behind Theorem \ref{thm: rec positive}, we need to introduce two new key terms: the linkage disequilibrium term $LD_2$ and the {\em flat variance}.
We define  $LD_2(\bfphi)$  to be the decrease in variance produced by recombination:
\[
LD_2(\bfphi) = V(\bfphi) - V(\Rec(\bfphi)).
\]
$LD_2$ can be shown to be equal to the covariance term $\sum_{i\neq j} E(X_i X_j)-E(X_i)E(X_j)$.
Theorem \ref{thm: rec positive} states that $LD_2(\bfphi)$ is negative at all stages of the process, unless the population is at linkage equilibrium,  in which case $LD_2(\bfphi)=0$.


A more geometric way of understanding $LD_2$ is through the notion of flat variance.
Let $\bfM=(E(X_1),E(X_2),...,E(X_\ell))\in \RR^\ell$; this vector represents the average individual in the population. 
The {\em global variance} of a population is defined as $GV(\bfphi)=E(\|X-\bfM\|^2)$.
Recombination does not affect the global variance, $GV(\bfphi)$, at all.
However, it changes the shape of the population by increasing the variance in the direction that is useful for selection, namely the fitness variance.
Consider the diagonal line $d=\{(x_1,...,x_\ell)\in\RR^\ell:x_1=x_2=\cdots=x_\ell\}$ and its $(\ell-1)$-dimensional orthogonal complement $P=\{(x_1,...,x_\ell)\in\RR^\ell:x_1+x_2+\cdots+x_\ell=0\}$, and let  $\pi_d$ and $\pi_P$ be the projection functions onto $d$ and $P$ respectively.
Using that $F(X)$ is the inner product of $X$ and  $(1,1,...,1)$, one can show that: 
\[
V(\bfphi)=\ell\cdot \Var(\|\pi_d(X)\|).
\]
We define the {\em flat variance} of a population to be the variance of its projection onto $P$ multiplied by a correcting factor:
\[
FV(\bfphi)=\frac{\ell}{\ell-1}\cdot  E(\|\pi_P(X-\bfM)\|^2).
\]
Informally, $V(\bfphi)$ measures how \emph{tall} a population is along the vector $(1,1,...1)$, while $FV(\bfphi)$ measures how \emph{fat} it is.
The effect of recombination on variance and flat variance satisfies a simple formula:
\[
V(\Rec(\bfphi)) = FV(\Rec(\bfphi)) = \frac{V(\bfphi)+(\ell-1)FV(\bfphi)}{\ell},
\]
and hence  
\[
 LD_2 =  \frac{\ell-1}{\ell}\Big(V(\bfphi)-FV(\bfphi) \Big).
 \]
Thus, $LD_2$ being negative is equivalent to $FV$ being greater than $V$, or, more informally, the population being fatter than it is tall along $d$.
The dynamics of this interaction are explained in Figure 2, and the effects for unbounded and bounded domains are illustrated in Figures 3 and 4 respectively.

\begin{figure}[!ht] \label{process}
  \begin{adjustbox}{addcode={\begin{minipage}{\width}}
{\caption{This illustration shows the level curves for 2-locus sex (red) and asex (blue) populations which begin with the same Gaussian type distribution (loci distributions are i.i.d.\ with $\kappa_3=0$). The $x$ and $y$ axes indicate fitnesses at the first and second loci respectively, fitness increasing along the up-right diagonal.
The reason the last application of selection gives a greater increase in mean fitness for sex is that, {\em as opposed to asex, the sex process capitalises on the increase in flat variance due to mutation}.
 The first phase is selection which decreases the fitness variance (given $\kappa_3=0$), and  does not interact with flat variance since fitness (measured along the $\diagup$-diagonal) is the only factor influencing the ability of an individual $\bfx$  to survive selection, while flat variance is measured along the planes where fitness is constant (the $\diagdown$-diagonal).
Selection thus causes the flat variance to be greater than the variance, or equivalently, causes negative $LD_2$, as we show in Theorem \ref{thm: rec positive}.
Mutation then increases both flat variance and variance by the same amount.
Recombination, only occurring in the sex process, averages out the fitness variance and flat variance, increasing the variance and decreasing the flat variance, as seen by the rounded form of the level curves at that phase in the figure. 
Notice that recombination does not increase global variance; it just transforms the flat variance, which is useless for selection, into fitness variance.
Finally, selection now has an increased effect on the mean fitness of the sexual population due to the larger fitness variance produced by recombination.
For the asex population, mutation will keep on increasing the flat variance, but, in the absence of recombination, selection will not  capitalise on this growth.}
 \end{minipage}},rotate=90,center} 
  \includegraphics[scale=.215]{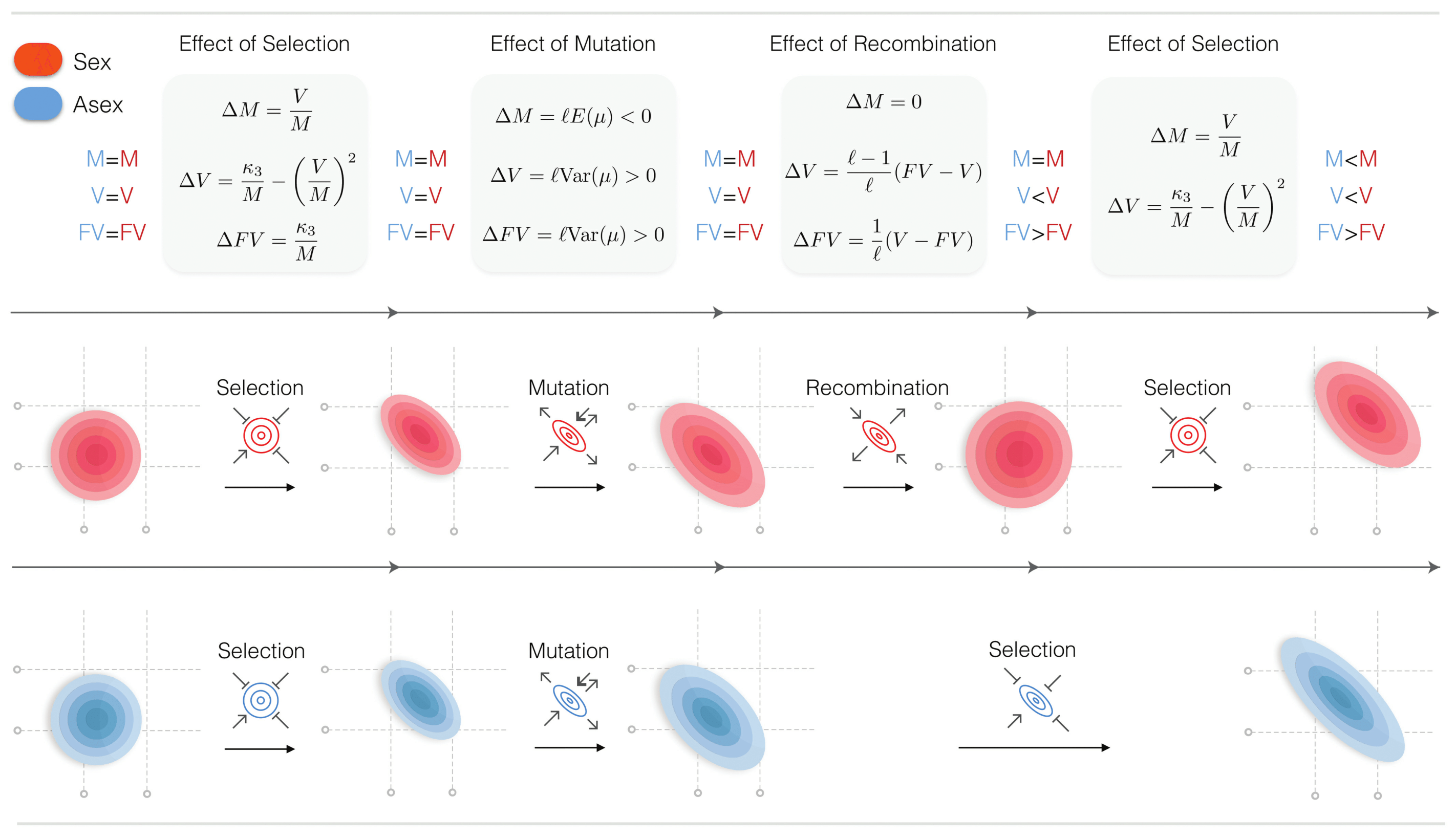}
  \end{adjustbox} 
\end{figure}

\begin{figure}[!ht]
  \begin{adjustbox}{addcode={\begin{minipage}{\width}}
  {\caption{The benefit of recombination. This figure shows the progress of two 2-locus populations (for the $\NN$ model), one of which is sexual and the other of which is \emph{initially} asexual. The $x$ and $y$ axes correspond to fitnesses at the first and second loci respectively,  the $z$ axis corresponding to probability density. Both populations begin with initial allele fitnesses of 5, and progress with mutation rate $10^{-4}$, the probability any given mutation is beneficial being $0.1$. During the first 2000 generations the sexual population quickly achieves greater mean fitness, and one can clearly see the increased flat variance of the asexual population. At stage 2000, a single application of recombination is made to the previously asexual population, converting that flat variance into variance and greatly increasing the rate of increase in mean fitness. Without any subsequent applications of recombination, however, the asexual population will eventually have smaller variance than the sexual one and will once again fall behind in mean fitness.   
 }\end{minipage}},rotate=90,center}
  \includegraphics[scale=.215]{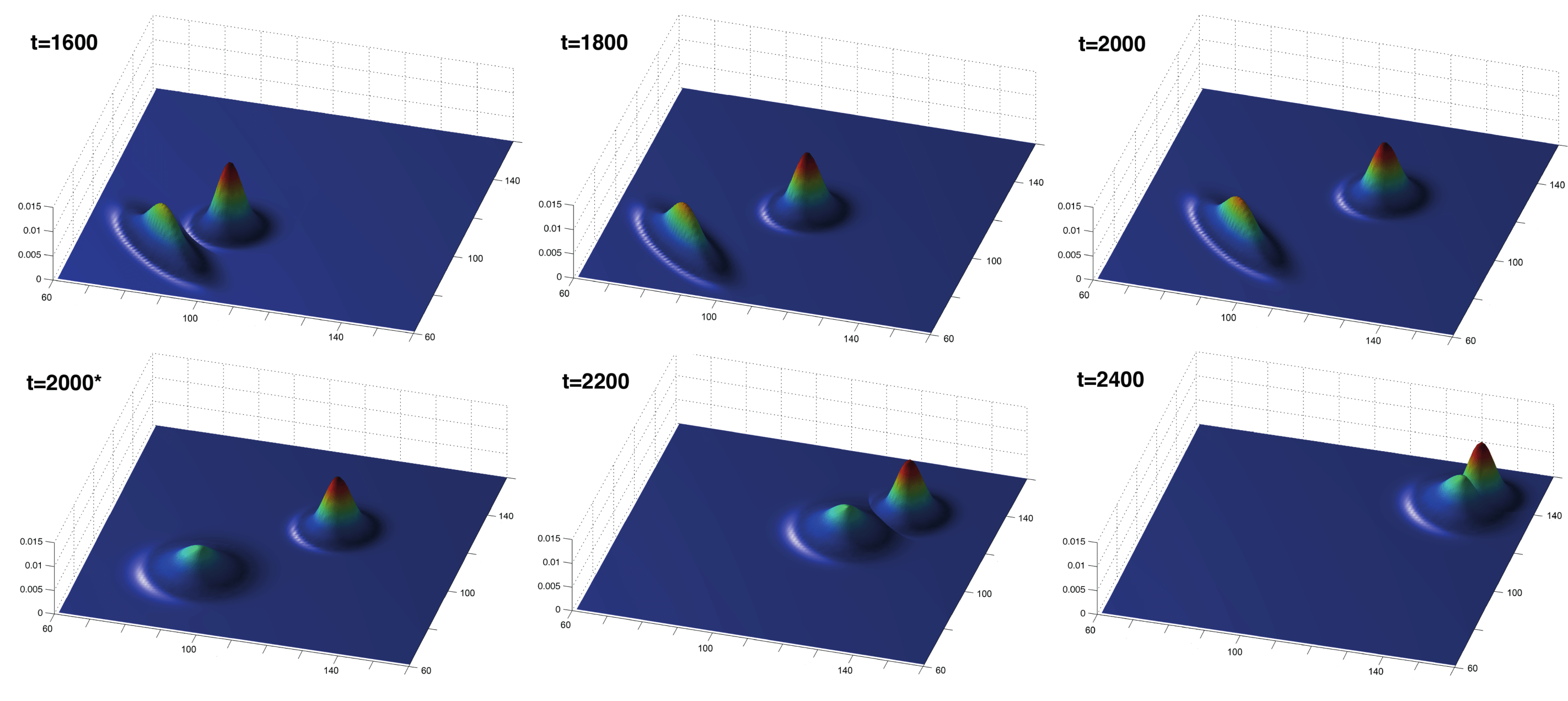}
  \end{adjustbox}
\end{figure}

\begin{figure} 
\includegraphics[scale=.2]{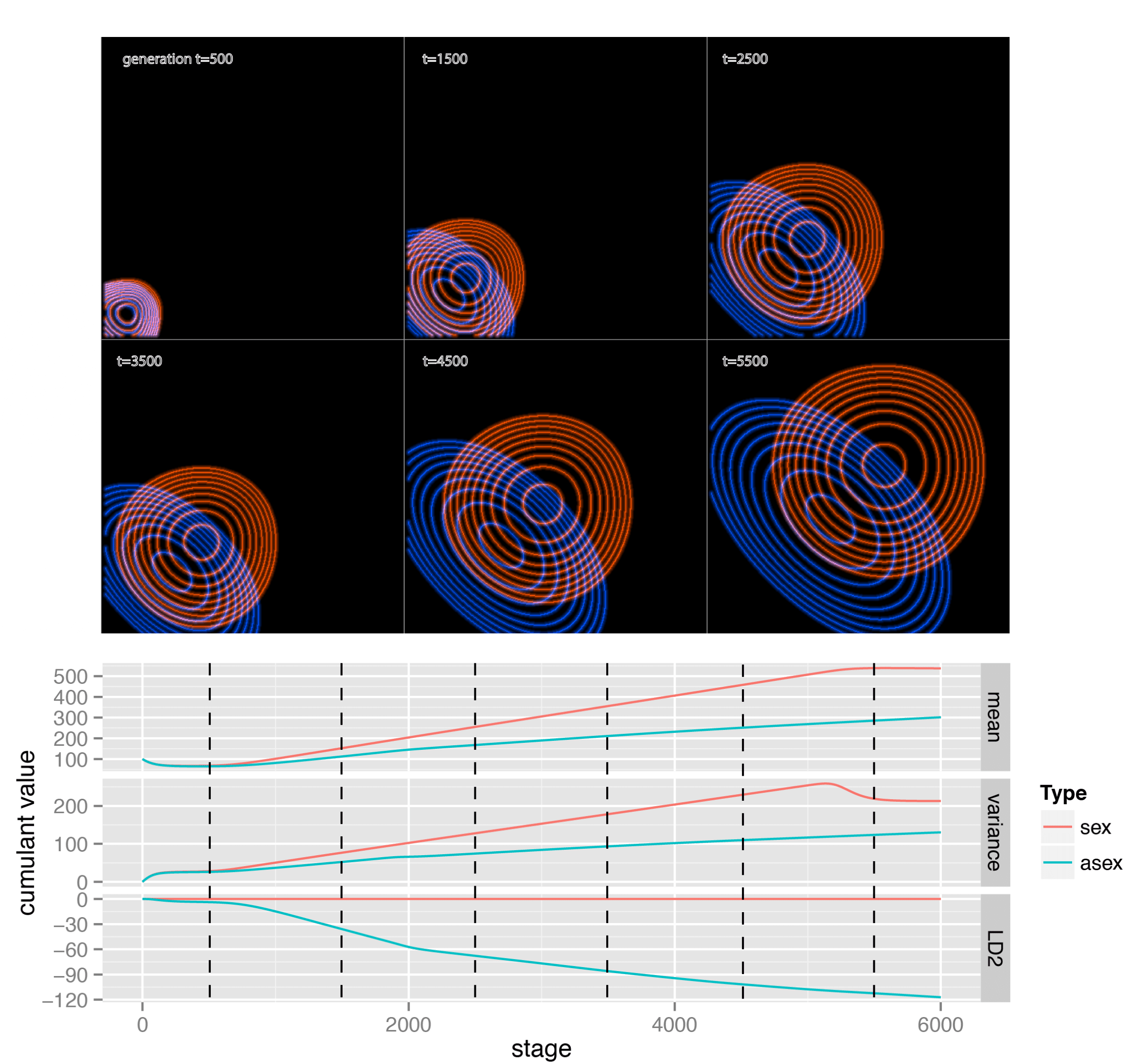}
\caption{This figure shows the level curves for 2-locus populations proceeding according to the bounded model, with maximum allele fitness 400,  mutation rate 0.2, and with the probability any given mutation is beneficial being $10^{-4}$. All alleles initially have fitness 50. The probability density level curves are depicted at stages 500, 1500, 2500, 3500, 4500 and 5500.
We can again observe the increase in flat variance and decrease in variance for the asexual population, and also that the sexual population does not necessarily have a higher global variance.}
\end{figure}


\

Theorems \ref{thm: effect rec} and \ref{thm: rec positive} show an important advantage that sex has over asex.
In comparing  sex and asex populations  evolving independently, however, these theorems do not suffice to entirely specify how the variances of the two populations differ at any given generation. To make this comparison we would need to understand the evolution of the  third cumulant,  which behaves differently in each process.
The evolution of the third cumulant depends on the fourth, which depends on the fifth, and so on.

Rather than analysing further the evolution of populations over time,  we now study what happens to the sexual and asexual populations in the long term.
We prove that, for the bounded model, whatever the initial populations are, sex outperforms asex in the long run.

We state the following theorem in terms of a mixed population containing both sexual and asexual individuals competing for resources.
Thus the population distribution $\bfphi$ now has domain $D\times \{\mathtt{s,a}\}$, the second coordinate indicating whether the individual is sexual or asexual.
Mutation acts exactly as before among each type of individual.
Selection is also the same, now using $M(\bfphi)=\sum_{\bfx\in D\times \{\mathtt{s,a}\}}F(\bfx) \bfphi(\bfx)$ to normalise.
Recombination  acts only among the sexual  individuals.

\begin{theorem}\label{main bounded alleles}
Given $\mu$ and $\ell$, for all sufficiently large bounds $N$ and for any initial population in which the proportion of sexual individuals is non-zero, the proportion of sexual individuals converges to 1 and the proportion of asexual ones converges to 0.
\end{theorem}

This is the longest and most complicated proof of the paper: the proof appears in (SI \S\ref{ss: theorem 3}).
Applying the Perron-Frobenius Theorem suffices to prove that the asex distribution converges to a limit, and using other techniques we are also able to find a good approximation for the mean fitness at that limit.
We do not prove that the sex process converges to a limit, but still get a good estimate of the geometrical average of the mean fitness over generations.
Such ideas would not work for the $\NN$- and $\ZZ$-models as in those cases the mean fitness diverges to infinity in both the sex and the asex processes. Figure 5 shows the manner in which sexual and asexual populations converge to their respective fixed points over time (while we do not prove that convergence to a fixed point always occurs for sexual populations, such convergence was observed in all simulations).

\begin{figure} 
\includegraphics[scale=.27]{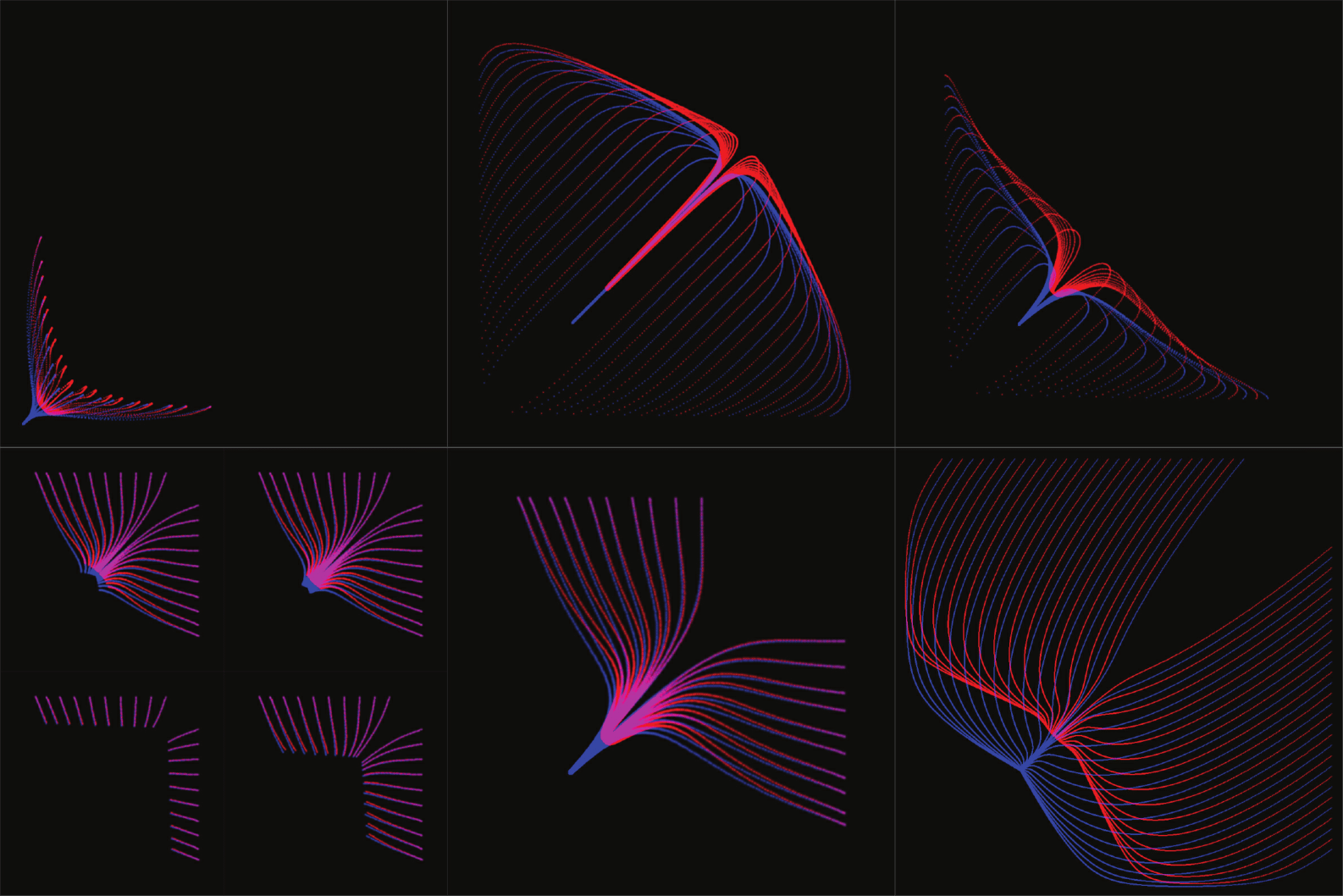}
\caption{Each of the six plots shows the trajectory of the centre of mass for various sexual and asexual 2-locus populations over multiple generations, for a number of different initial populations and for the bounded model.  Each point represents the centre of mass of a population at a single generation, and the populations were then allowed to evolve for sufficiently many generations that an equilibrium point was reached.
The bottom-left plot shows intermediate steps in the evolution towards the middle plot in the bottom row.
For that plot, we have 40 different initial populations, half sexual (red), half asexual (blue).  The bound, $N$, on gene fitness is 50 for all plots except for the top-centre and bottom-right, where $N=301$.
The probability of mutation is $0.5$ except for the top-left plot, where the probability of mutation is $0.9$.
The probability that a mutation is beneficial is $0.001$ in all cases.
Starting from the top-left and moving clockwise, the original populations are Gaussian distributions with standard deviations  $5$, $25$, $6$, $8$, $6$ and $6$ respectively. 
}
\end{figure}

\section{Discussion}
In nature one must surely expect a variety of mechanisms to be of significance in determining the most efficient methods of reproduction. As well as those factors mentioned in the introduction, sex may provide advantages for species not subject to random mating by strengthening selection \cite{GY}, for example, or may provide a straightforward advantage in providing two parents to care for young offspring \cite{DL}. Such arguments, however, do not suffice to explain the prevalence of sex in species for which random mating is a good approximation or without parental care. Our aim here has been to rigorously establish a fundamental and underlying mechanism conferring strong advantages to sex. 
We have seen that independence between loci allows for more rapid growth in mean fitness. 
In the absence of such independence, \emph{the selection of fitter alleles at a particular locus will be stronger when other genes have lower fitness values}. A simple analogy may be given in terms of the comparative value of improvements to sensory abilities: If an organism has little sight, a small improvement in hearing may be more important than it is for an organism with excellent vision.
Thus, in the asex process, the result is that individuals which have high fitness on a gene, tend to have low fitness on another -- this is essentially what negative $LD_2$ means, and what is behind the proof of Theorem \ref{thm: rec positive}. The effect of the sex process is to break down these negative associations, but not to increase or decrease the \emph{global} variance of a population. The key role of recombination is to transform the variance produced by negative associations --  the  flat variance -- into the form of variance which can then be acted upon by selection -- the fitness variance.

\

Of course a natural question, having considered the infinite populations case, is the extent to which this analysis carries over to the finite populations model. The principal difference in moving to finite populations is that the process is no longer deterministic.  The equations governing the change in mean fitness and variance due to selection and mutation for the infinite population model would now perfectly describe the \emph{expected} effect of mutation and selection for finite populations, and the finite populations model could be seen simply as a stochastic approximation to the infinite case, were it not for the loss in variance and higher cumulants due to sampling 
(since picking $n$ individuals from a distribution with variance $v$ produces a population with expected variance $v(n-1)/n$).
 For large populations this effect will be very small on a stage by stage basis, and so our analysis for infinite populations can be seen as a good approximation over a number of generations which is not too large. Ultimately, however, sampling will have the effect that mean fitness for the population no longer increases without limit: once variance is sufficiently large the expected loss in variance due to sampling balances the increase that one would see for an infinite population with the same cumulants. Larger populations are thus able to sustain much higher mean fitnesses than small ones. 
 
 \

While sexual reproduction has been seen here to confer strong advantages in the absence of epistasis, i.e.\ in the setting of simplistic and entirely modular fitness landscapes, we have said nothing about how this picture changes in the presence of epistasis. Assuredly, the task of efficiently navigating fitness landscapes (i.e.\ optimisation) is one that, beyond its relevance here, is of fundamental significance across large areas of applied mathematics and computer science (hence the recent interest of computer scientists in the benefits of sexual reproduction\cite{LP}). However large the role of epistasis in the biological context, it is certainly true that in most of these applications epistasis (in one guise or another) plays a crucial role, and so the interesting question becomes that as to whether sexual reproduction continues to offer these substantial benefits in the face of more complex fitness landscapes. It may be the case that as well as capitalising more efficiently on existing modularity, sex plays a fundamental role in \emph{finding} modularity\cite{LP2}. One would expect a proper analysis to require classification of fitness landscapes in terms of their amenability to different forms of population based search (see, for example, the work of Prugel-Bennet\cite{PB}).

%
%
%


\newpage

\noindent $^1$ Andrew E.M. Lewis-Pye, Department of Mathematics, London School of Economics, London UK. Email: andy@aemlewis.co.uk. \\

\noindent $^2$ Antonio Montalb\'{a}n, Department of Mathematics, University of California, Berkeley, USA. Email: antonio@math.berkeley.edu.




%

\newpage

\section{Extended Data} \label{EXTDAT}

\noindent Figures 6 and 7 display the outcome of simulations for the additive finite populations model. Figures 8 and 9 display the outcome of simulations for the additive infinite populations model. Figure 10 displays the outcome of simulations for the finite populations multiplicative model. Where required for our proofs, we have shown that the proportion of a population at the boundaries will generally be small after sufficiently many generations have passed. Tables 1, 2 and 3 show the proportion of the population at the boundaries for the additive infinite populations $\NN$-model and also for the bounded model.  All variants of the model are described in \S\ref{variants}. 

\begin{figure}[h!] \label{finitegraphs}
\includegraphics[scale=.38]{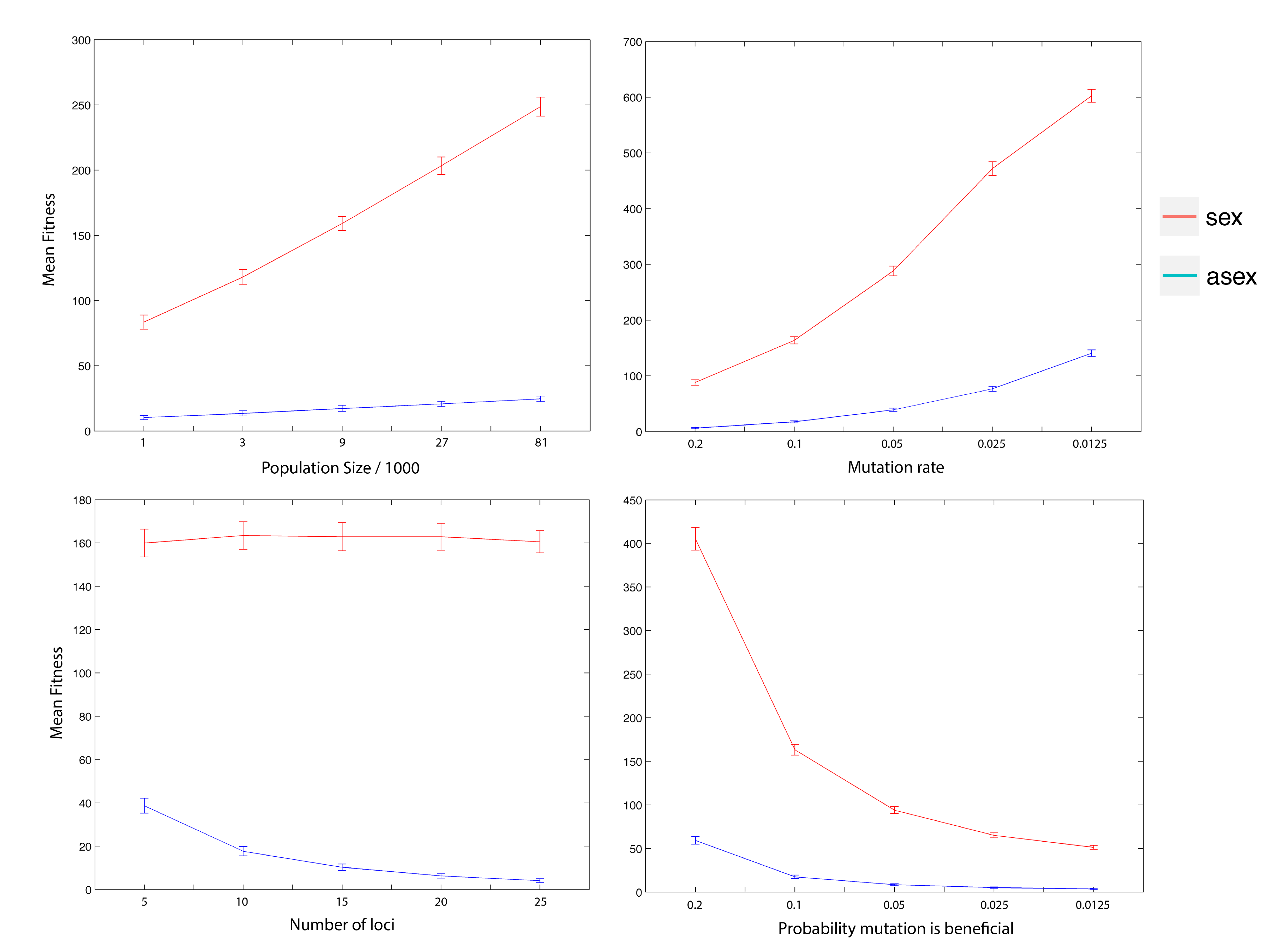}
\caption{Simulations for the finite additive model. In these simulations the `standard' input parameters were: population size 10000; mutation rate 0.1; probability mutation is positive 0.1; 10 loci, initial gene fitness 5. In each graph one parameter is varied, while the other parameters take the standard values. 100 simulations were run for each parameter set, and the mean fitnesses as well as the standard deviations for these mean fitnesses are depicted, after a number of generations which is sufficient for the mean fitness to stabilise. This number of generations was taken to be 4000, except for the case of varying mutation rate where 20000 generations were run for each simulation.
}
\end{figure} 

\newpage

\begin{table}[h] \label{zeros0}
\[
\begin{array}{ c || c | c | c }
                           	&	\mbox{initial gene fitness =}1 						&  10		&  20      	\\	
\hline			\hline

0.2				&	3.5\times 10^{-57} \ / \ 1.0\times 10^{-31} 	& 1.1 \times 10^{-57} \ / \ 3.2\times 10^{-32} 	& 3\times 10^{-58} \ / \ 5.6 \times 10^{-33}    \\
\hline 
		
0.1			&	1.6\times 10^{-70} \ / \ 1.0 \times 10^{-39}	&	3.7 \times 10^{-71}\ / \ 1.5\times 10^{-40} 	& 8.2 \times 10^{-72}\ / \ 6.6\times 10^{-42}   	\\		
\hline

0.05		&     2.9\times 10^{-82}\ / \ 6.6\times 10^{-45} 	&	5.7\times 10^{-83}\ / \ 3.6 \times 10^{-46}   &  1.1\times 10^{-83} \ / \ 2.4 \times 10^{-48} 	\\ 			
\hline

0.025	&	2.7\times 10^{-92} \ / \ 1.2\times 10^{-47} 			& 	4.6\times 10^{-93} \ / \ 1.7 \times 10^{-49}	&	 1.0\times 10^{-93} \ / \ 9.4 \times 10^{-53}  			 \\
\hline 
0.0125 & 9.8\times 10^{-101} \ / \ 1.5\times 10^{-48}  & 1.5\times 10^{-101}\ / \ 4.2 \times 10^{-51}  & 4.1\times 10^{-102} \ / \ 8.0\times 10^{-56} \\
\end{array}
\]
\caption{The table concerns the infinite additive $\NN$-model, and shows the proportion of a 2-locus population which has fitness 1 at either locus for sex/asex, after 1000 generations, for varying initial gene fitnesses, and for varying mutation rates. In all cases the probability that a given mutation is beneficial is $10^{-1}$.  }  
\label{table: Deltas}
\end{table}

\begin{table}[h]  \label{zeros1}
\[
\begin{array}{ c || c | c | c }
                           	&	N =200 						&  300	&  400      	\\	
\hline			\hline

0.2				&	3.4\times 10^{-57} \ / \ 3.8\times 10^{-49} 	& 1.5 \times 10^{-85} \ / \ 2.1\times 10^{-73} 	& 6.6\times 10^{-114} \ / \ 1.2 \times 10^{-97}    \\
\hline 
		
0.1			&	3.5\times 10^{-100} \ / \ 5.0 \times 10^{-95}	&	3.8 \times 10^{-150}\ / \ 2.1\times 10^{-142} 	& 3.9 \times 10^{-200}\ / \ 8.9\times 10^{-190}   	\\		
\hline

0.05		&     1.4\times 10^{-150}\ / \ 1.1\times 10^{-147} 	&	7.5\times 10^{-226}\ / \ 1.8 \times 10^{-221}   &  3.9\times 10^{-301} \ / \ 2.8 \times 10^{-295} 	\\ 			
\hline
\end{array}
\]
\caption{The table concerns the infinite additive bounded model, and shows the proportion of a 2-locus population which has fitness 1 at either locus for sex/asex, after 25000 generations, for varying $N$ (maximum allele fitness), and for varying mutation rates. In all cases the probability that a given mutation is beneficial is $10^{-3}$ and all alleles initially have fitness 50.}  
\label{table: Deltas}
\end{table}

\begin{table}[h]  \label{zeros2}
\[
\begin{array}{ c || c | c | c }
                           	&	N =200 						&  300	&  400      	\\	
\hline			\hline

0.2				&	6.5\times 10^{-29} \ / \ 1.2\times 10^{-30} 	& 3.4 \times 10^{-42} \ / \ 9.3\times 10^{-45} 	& 2.6\times 10^{-55} \ / \ 1.1 \times 10^{-58}    \\
\hline 
		
0.1			&	9.2\times 10^{-16} \ / \ 2.5 \times 10^{-16}	&	7.1 \times 10^{-23}\ / \ 1.1\times 10^{-23} 	& 7.5 \times 10^{-30}\ / \ 6.3\times 10^{-31}   	\\		
\hline

0.05		&     1.5\times 10^{-8}\ / \ 1.0\times 10^{-8} 	&	2.3\times 10^{-12}\ / \ 1.3 \times 10^{-12}   &  4.5\times 10^{-16} \ / \ 2.2 \times 10^{-16} 	\\ 			
\hline
\end{array}
\]
\caption{The table concerns the infinite additive bounded model, and shows the proportion of a 2-locus population which has fitness $N$ (maximum allele fitness) at either locus for sex/asex, after 25000 generations, for varying $N$ and for varying mutation rates. In all cases the probability that a given mutation is beneficial is $10^{-3}$ and all alleles initially have fitness 50.}  
\label{table: Deltas}
\end{table}

\newpage

\begin{figure}[h!] \label{finitegraphs2}
\includegraphics[scale=.75]{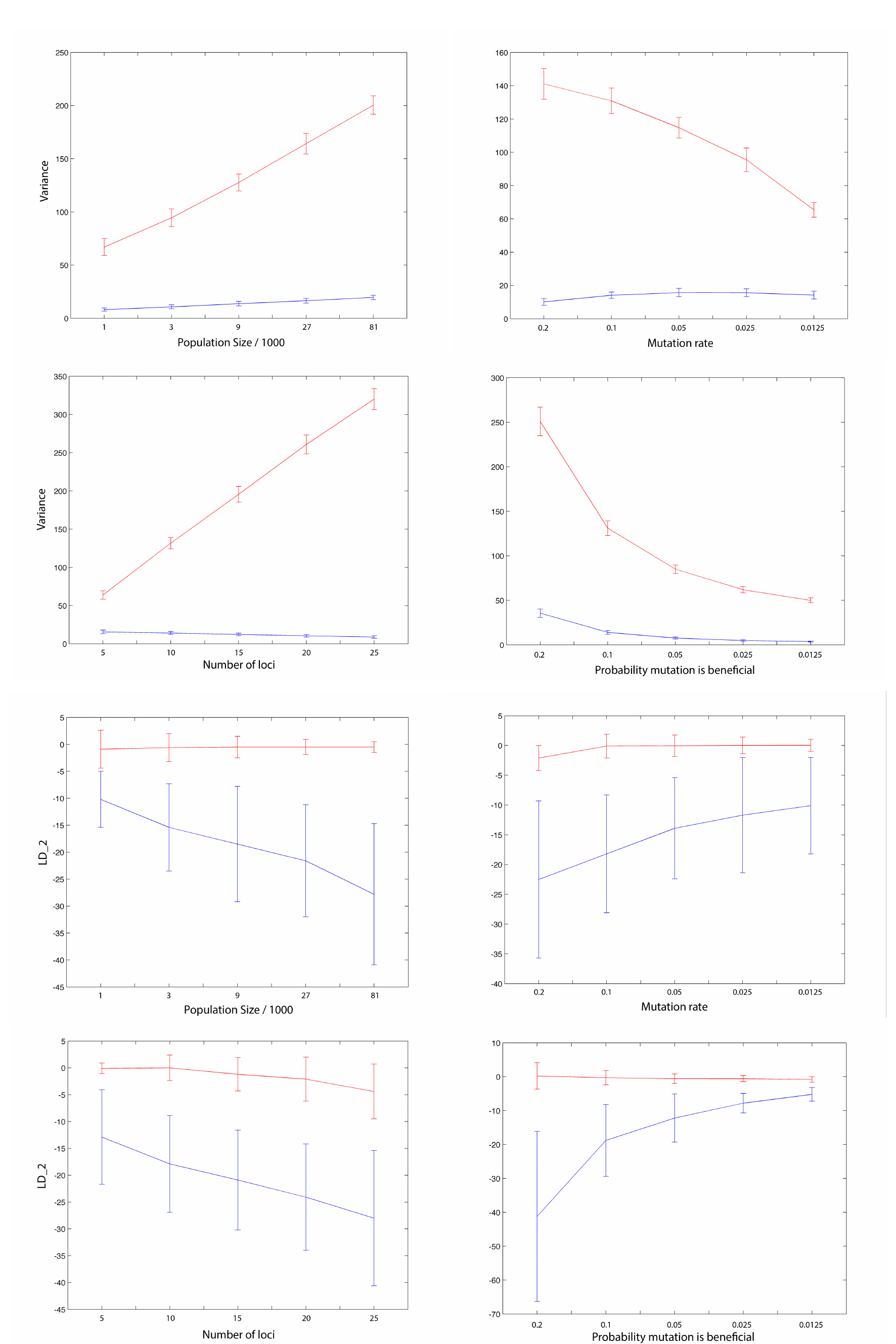}
\caption{These graphs display variance and $LD_2$  for the same simulations which have their mean fitnesses displayed in Figure 6.}
\end{figure} 

\newpage 

\begin{figure}[h!] \label{VoverM}
\includegraphics[scale=.165]{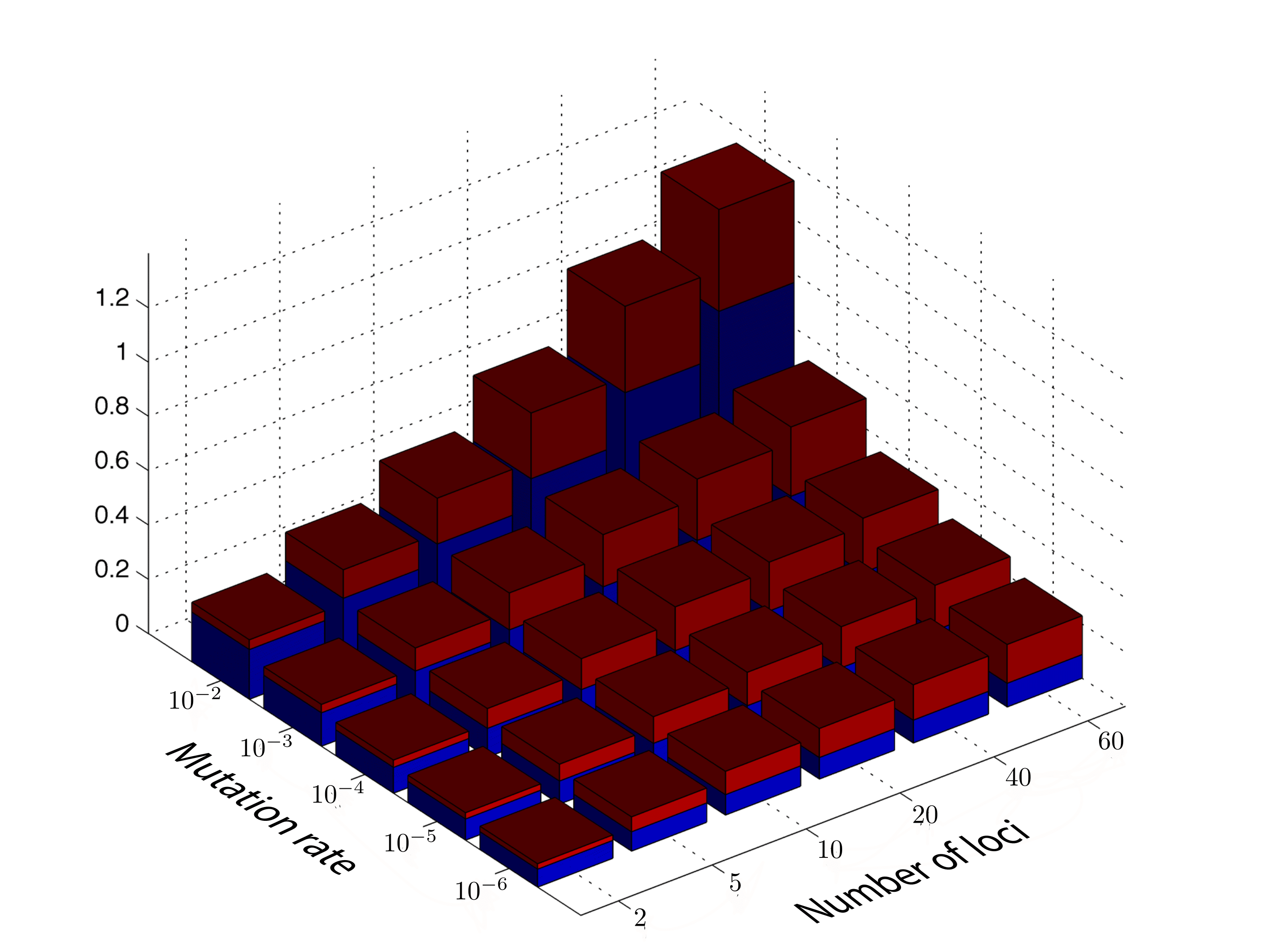}
\caption{Simulations for the infinite additive model appear to show $V/M$ reaching a limit value over time. The figure shows approximate values for these limits, for sex (red) and asex (blue). In all these simulations the probability that a given mutation is beneficial was fixed at 0.1, and gene fitnesses were initially 5 (although the latter parameter has no effect on the limit values found).}
\end{figure} 

\newpage 

\begin{figure}[!ht] \label{Fig9}
  \begin{adjustbox}{addcode={\begin{minipage}{\width}}
  {\caption{Simulations for the additive infinite model. Each set of four graphs shows (a) $M$, (b) $V/M$, (c) $\kappa_3/M$, (d) $\kappa_4/V$ for one simulation. In all simulations initial gene fitnesses are 5, and each simulation is then specified by a triple: $\ell$ specifies the number of loci, $p$ is the probability of mutation, $q$ is the probability a given mutation is beneficial. 
 }\end{minipage}},rotate=90,center}
  \includegraphics[scale=.58]{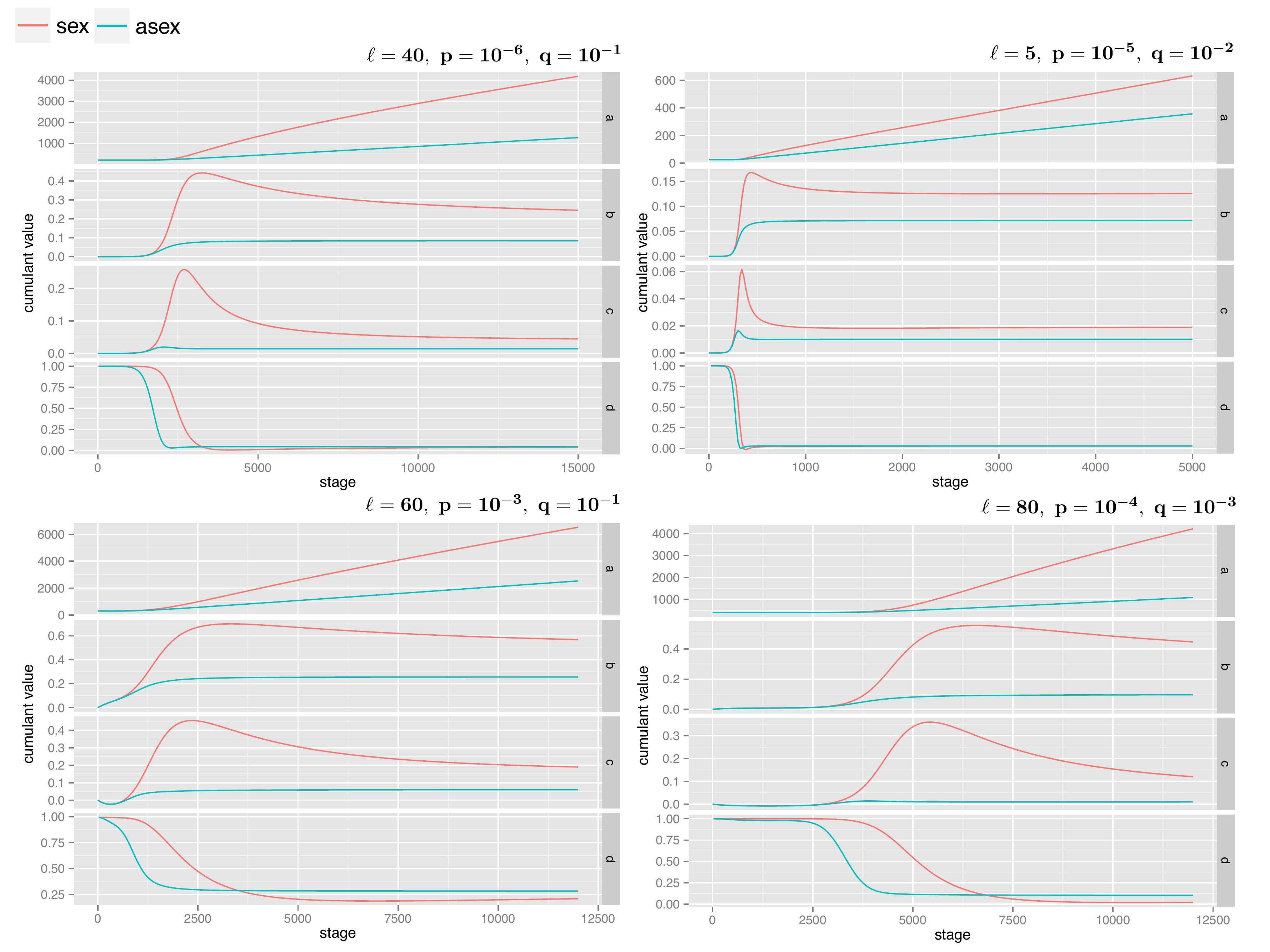}
  \end{adjustbox}
\end{figure}

\newpage

\begin{figure}[!ht] \label{finitegraphsmult}
\includegraphics[scale=.45]{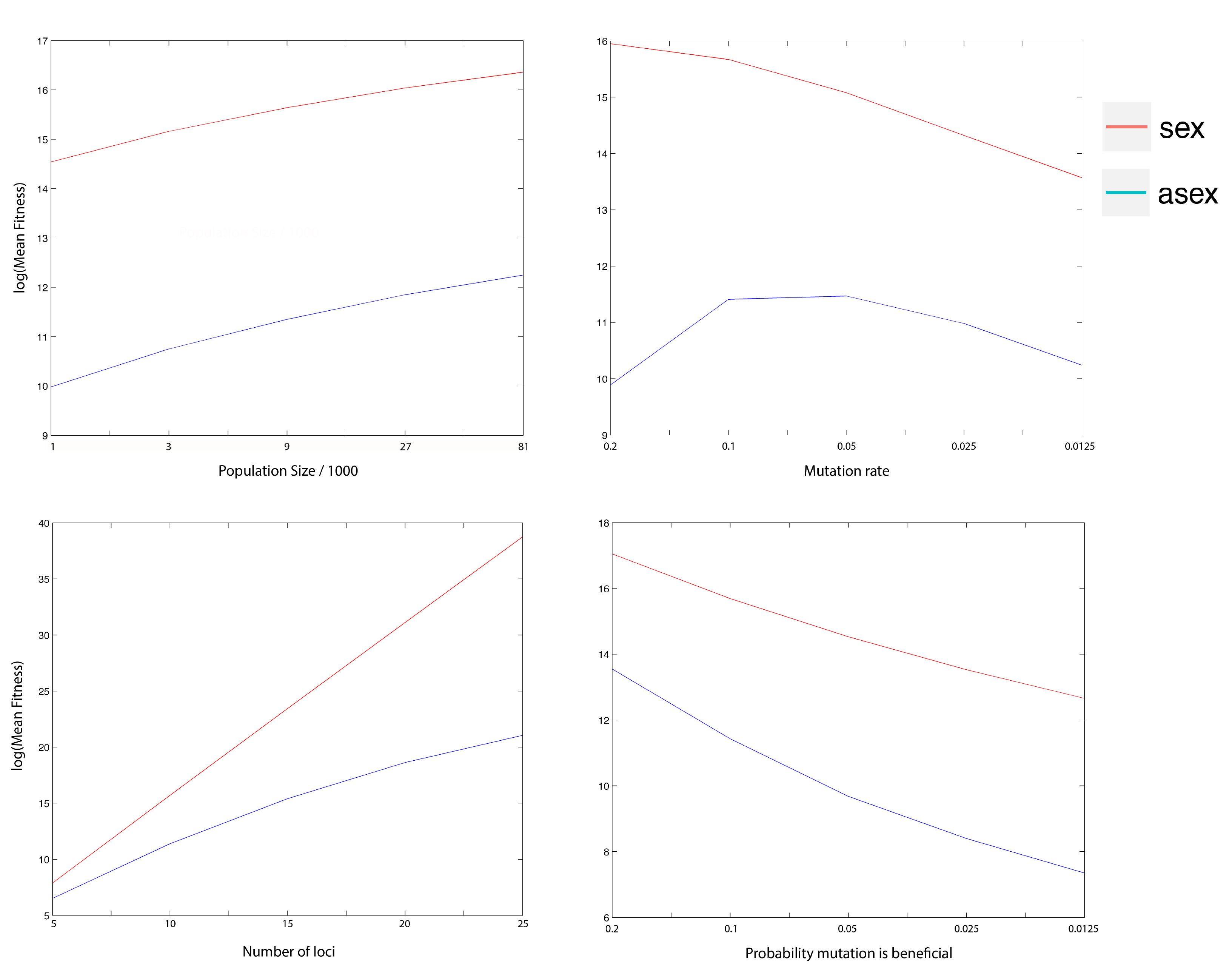}
\caption{Simulations for the finite multiplicative model. In these simulations the `standard' input parameters were: population size 10000; mutation rate 0.1; probability mutation is positive 0.1; 10 loci, initial gene fitness 1. In each graph one parameter is varied, while the other parameters take the standard values. 100 simulations were run for each parameter set, and the logarithms (base 10) of the mean fitnesses  are depicted after 500 generations (without any suggestion that the mean fitness has stabilised by this point).}
\end{figure} 

\newpage

\section{Supplementary Information}

In this section we provide a much deeper analysis of all the claims made in the main article.
We start, in Subsection \ref{key values}, by proving Theorem 1 and showing how the various operations affect different features of a population, such as mean fitness and variance.  
Subsection \ref{ss: theorem 2} is dedicated to the proof of Theorem \ref{thm: rec positive}.
Although Theorem \ref {thm: effect rec} and  Theorem \ref{thm: rec positive} are both concerned with establishing negative values for $LD_2$, their proofs are completely different and give us alternative ways of understanding the model.
Theorem \ref{main bounded alleles} is proved in Subsection \ref{ss: theorem 3}.
The proof of this theorem is much longer than those of the previous theorems, and, again, it is very different in style. 
Once we have proved our main theorems, we move on to discuss variants of our model in Subsection \ref{variants}.
The variants we consider are the finite version and the multiplicative version.
We do not have a full mathematical analysis for those models, but present the results of simulations.
The outcomes of simulations are presented in the previous section, Extended Data  \S\ref{EXTDAT}.

\subsection{The evolution of the key values} \label{key values} 

In this subsection we review some well known facts and describe in more detail how mutation, selection and recombination affect mean fitness, variance, $LD_2$ and flat variance, proving the claims  made in Section \ref{se: Basic Analysis}. 
The objective is to establish all of the results in the table below:

\begin{table}[h] 
\[
\begin{array}{ c || c | c | c }
\mbox{Effect of:}		&	\mbox{selection} 						&  \mbox{mutation}		&  \mbox{recombination}		\\			\hline			\hline
\Delta M				&		V/M								&	\ell\ E(\mu) 	&		0				 \\ \hline
\Delta V				&	\kappa_3/M-(V/M)^2						&	\ell\ \Var(\mu) 	&		-LD_2 			\\			\hline
\Delta LD_2			&     - \sum_{i\neq j} V_iV_j/M^2  \ \ \ \ \ \ (\ast)& 0	&	-LD_2		\\ 			\hline
\Delta \kappa_3			&	(V/M)((\kappa_4/V)-3(\kappa_3/M)+2(V/M)^2)									& 	\ell\  \kappa_3(\mu)	&		-LD_3			 \\
\end{array}
\]
\caption{By $\Delta M$ is meant the change in $M$ produced by the relevant operation. 
All  values ($M$, $V$, etc) inside the table are with respect to the population before the relevant operation is applied:
the box stating that $\Delta M$ for selection is $V/M$ should be read $M(\Sel(\bfphi))-M(\bfphi)=V(\bfphi)/M(\bfphi)$.
$(\ast)$  The stated effect of selection on $LD_2$ is only valid in the case that selection is acting on a population at linkage equilibrium.}  
\label{table: Deltas}
\end{table}
\subsubsection{The evolution of mean fitness and variance} 

The impact of mutation on the mean, variance and all cumulants is simply described (recall that mean fitness and variance are the first and second cumulants of $F(X)$). 
If $Y$ and $C$ are independent random variables and $\kappa_n $ is the $n$th cumulant, then $\kappa_n(Y+C)=\kappa_n(Y)+\kappa_n(C)$. 
Thus the effect of mutation on the mean fitness is to increase it by $\ell E(\mu)$ (which will be negative given our assumptions on $\mu$). Similarly, the effect on variance is to increase it by $\ell \Var(\mu)$.      

The effect of selection is given by the following lemma (while these claims are either well known or easily established, we include a proof for the sake of completeness).  Here  $M=M(\bfphi)$, $V=V(\bfphi)$, $\kappa_3=\kappa_3(\bfphi)$, $\kappa_4=\kappa_4(\bfphi)$ and $M^{\ast}$, $V^{\ast}$, $\kappa_3^{\ast}$ are the corresponding values for $\bfphi^{\ast}=\Sel(\bfphi)$. 

\begin{lemma}
The effect of selection on mean fitness, variance and $\kappa_3$ is given by:
\begin{eqnarray*} 
M^\ast -M  &=& V/M, \\ 
V^{\ast} -V &=& \kappa_3/M-(V/M)^2, \\
\kappa_3^{\ast}-\kappa_3 &=& (V/M)((\kappa_4/V)-3(\kappa_3/M)+2(V/M)^2). 
\end{eqnarray*} 
\end{lemma}
\begin{proof}
We prove the first two identities. The third then follows with a little more algebraic manipulation, by almost identical methods. 
In order to see the first identity, note that: 
\[ M^{\ast} = \sum_{\bfx} F( {\bfx})\ \Sel(\bfphi)(\bfx) =  \frac{1}{M}\sum_{\bfx} F(\bfx)^2 \bfphi(\bfx). \]
\noindent  Now, using that the second moment about the origin, $\sum_{\bfx} F(\bfx)^2 \bfphi(\bfx)$, is equal to $V+M^2$ we get:  
\[ V=\Big( \sum_{\bfx} F(\bfx)^2 \bfphi(\bfx)\Big)   \  - M^2= M^{\ast} M-M^2. \]  
\noindent This gives the well known identity  $ V/M= M^\ast -M$, as required. In order to derive the second identity, we recall the formula for the third central moment:

\begin{eqnarray*}
 \kappa_3 &=&   \sum_{\bfx} (F(\bfx)-M)^3  \bfphi(\bfx)\\
 &= & \sum_{\bfx} (F(\bfx)^3-3F(\bfx)^2M+3F(\bfx)M^2-M^3)\ \bfphi(\bfx) \\
 &= & \Big( \sum_{\bfx} F(\bfx)^3 \bfphi(\bfx)\Big)  -3M\left( \Big(\sum_{\bfx} F(\bfx)^2 \bfphi(\bfx)\Big)-M^2 \right) -M^3 \\
 &= & \Big( \sum_{\bfx} F(\bfx)^3 \bfphi(\bfx)\Big)  -3MV -M^3.
 \end{eqnarray*} 

\noindent Then: 
\[ V^{\ast}= \Big( \sum_{\bfx} F(\bfx)^2\  \bfphi^*(\bfx)\Big) - (M^*)^2= \frac{1}{M} \Big( \sum_{\bfx} F(\bfx)^3 \bfphi(\bfx)\Big) \  -(M^{\ast})^2. \] 
\noindent Substituting $V/M+M$ for $M^{\ast}$, we get:
 \[ V^{\ast}-V=   \frac{1}{M} \Big( \sum_{\bfx} F(\bfx)^3 \bfphi(\bfx)\Big) -\left(M^2 +2V +V^2/M^2\right)-V = \kappa^3/M - V^2/M^2,  \]
 \noindent as required. 
 \end{proof}
 
 Let us now consider recombination.
 Recall that $X_i$ is a random variable taking values according to the distribution $\phi_i$ (specifying the distribution at the $i$th locus). 
 At any point, the mean is given by $M(\bfphi)=\sum_i E(\phi_i)$. Since recombination has no effect on each $\phi_i$, it also has no impact on $M(\bfphi)$.
 The change of variance due to recombination is $-LD_2$, by the definition of $LD_2$.

 \subsubsection{Recombination, $LD_2$ and flat variance} \label{sss: LD2}

 It is not the case, however, that the sum of variances is the variance of the sum, unless the random variables in question are independent.
 Since $\Rec(\bfphi)$ is at linkage equilibrium, and recombination leaves each $\phi_i$ unchanged, we conclude that $V(\Rec(\bfphi))=\sum_i V_i$, where $V_i=\Var(\phi_i)$.
 We therefore have:
 \[
 LD_2 = V - \sum_{i=1}^\ell V_i.
 \]
 
 There is a second way to calculate $LD_2$ that will be useful later.
 The variance can be expressed: 
 \begin{eqnarray*}
 V(\bfphi) &=& E\Big(\big(\sum_{i=1}^\ell X_i\big)^2\Big)-\Big(E\big(\sum_{i=1}^{\ell} X_i\big)\Big)^2 \\
 &=& \sum_{i=1}^{\ell} \Big( E(X_i^2)-E(X_i)^2\Big) + \sum_{i\neq j} \Big( E(X_iX_j)-E(X_i)E(X_j) \Big). \\
 &=& \Big( \sum_{i=1}^{\ell} V_i \Big) +\sum_{i\neq j} \Big( E(X_iX_j)-E(X_i)E(X_j) \Big).    
 \end{eqnarray*}
 
\noindent This gives a description of $LD_2$ as a sum of covariance terms:
 \[
 LD_2 = \sum_{i\neq j} \Big( E(X_iX_j)-E(X_i)E(X_j) \Big).
 \]
 
Letting $\bfM=(E(X_1),E(X_2),...,E(X_\ell))\in \RR^\ell$, recall that the global variance $GV=GV(\bfphi)$ was defined as $E(\|X-\bfM\|^2)$, and the flat variance $FV=FV(\bfphi)$ was defined as $ E(\|\pi_P(X-\bfM)\|^2) \ell/(\ell-1)$.
Using that $\pi_P(X)+\pi_d(X)=X$ and that, by Pythagoras, $\|\pi_P(X-\bfM)\|^2+\|\pi_d(X-\bfM)\|^2=\|X-\bfM\|^2$, we get:
\begin{multline*}
GV = \sum_{i=1}^\ell V_i = E(\|X-\bfM\|^2) = \\
E(\|\pi_P(X-\bfM)\|^2) + E(\|\pi_d(X-\bfM)\|^2) = \frac{\ell-1}{\ell}FV+\frac{1}{\ell}V.
\end{multline*}
Thus, since $\sum_{i=1}^\ell V_i$ is  unaffected by recombination, so is $((\ell-1)FV+V)/\ell$.
We can also deduce that if $\bfphi^*$ is at linkage equilibrium and $V^*=\sum_{i=1}^\ell V_i^*$, then $FV^*=V^*$.
It follows that the effect of recombination on $V$ and $FV$ is to make them equal while leaving $((\ell-1)FV+V)/\ell$ unchanged, thus making them both equal to $((\ell-1)FV+V)/\ell$.
We then have: 
\[
V(\Rec(\bfphi))-V(\bfphi) = \frac{\ell-1}{\ell}(FV-V) 
\quad \mbox{ and } \quad
FV(\Rec(\bfphi))-FV(\bfphi) = \frac{1}{\ell}(V-FV), 
\]
and  
\[
LD_2 = ((\ell-1)/\ell) (V-FV).
\]

\subsubsection{The evolution of $LD_2$}

The most direct way in which recombination affects mean fitness is by changing the variance, which then affects the growth in mean fitness via selection.
The change in variance due to recombination is given by $-LD_2$.
Thus, to show that recombination has a positive effect on variance, one must show that $LD_2$ is negative. 
In this subsection we analyse the effect on $LD_2$ given by the different operations.
As part of our analysis we get a proof of Theorem \ref{thm: effect rec}.

 Since $LD_2=0$ when at linkage equilibrium, we have $LD_2(\mathtt{Rec}(\bfphi))=0$.

Mutation has no effect at all on $LD_2$ as shown by the following lemma.

\begin{lemma}
For any population $\bfphi$, $LD_2(\Mut(\bfphi))=LD_2(\bfphi)$.
\end{lemma}
\begin{proof}
Recall the definition of mutation in terms of the random variables $C_i$.
\begin{eqnarray*}
LD_2(\Mut(\bfphi))     &=&	 \sum_{i\neq j} \Big( E((X_i+C_i)(X_j+C_j))-E(X_i+C_i)E(X_j+C_j) \Big)   \\
				 &=&	 \sum_{i\neq j} \Big( E(X_iX_j) + E(X_iC_j) + E(C_iX_j) + E(C_iC_j) \\
				 & & \quad -E(X_i)E(X_j)-E(X_i)E(C_j)-E(C_i)E(X_j)-E(C_i)E(C_j) \Big).  \\
\end{eqnarray*}
Since $C_i$ and $C_j$ are independent, and are independent of $X_i$ and $X_j$,  most of these terms cancel, leaving $E(X_iX_j)-E(X_i)E(X_j)$ as required.
\end{proof}

Let us take this opportunity to note that mutation has no effect at all on linkage equilibrium: 
This is because if the variables $X_i$ are independent, so are the variables $X_i+C_i$.
Also, since $FV=V-(\ell/(\ell-1))LD_2$, we conclude that the effect of mutation on flat variance is the same as that on variance: 
 $FV(\Mut(\bfphi))-FV(\bfphi) = \ell \Var(\mu)$.

\

The effect of selection on $LD_2$ is more complex and is given by Theorem \ref{thm: effect rec} (restated below) in the case that the operation is applied to a population at linkage equilibrium.
The rest of the subsection is dedicated to proving it.
We define $LD_3$ to be the decrease in the third cumulant of $F(X)$ produced by recombination.
Thus,
\[
LD_3(\bfphi)= \kappa_3(F(X)) - \sum_{i=1}^\ell \kappa_3(X_i).
\]
As with the other values, we use $\kappa_3$ to denote $\kappa_3(F(X))$ and $\kappa_{3,i}$ to denote $\kappa_3(X_i)$.

\vspace{0.2cm} 

\noindent \textbf{Theorem 1}. 
\emph{If $\bfphi^{\ast}=\Sel(\bfphi)$ was obtained by an application of selection to a population $\bfphi$ at linkage equilibrium, then 
the effect of recombination on variance is given by:
\[
V(\Rec(\bfphi^{\ast}))-V(\bfphi^{\ast}) = \frac{\sum_{i\neq j} V_iV_j}{M^2},
\]
where $V_i=\Var(\phi_i)$ and $M=M(\bfphi)$.}


\vspace{0.2cm} 

Theorem \ref{thm: effect rec} asserts, in other words, that $LD_2(\Sel(\bfphi))=  -(\sum_{i\neq j} V_iV_j)/M^2$ if $\bfphi$ is at linkage equilibrium.  The key to the proof is to study the effect of selection on each locus separately, as given by the following lemma.
Let us describe our notation. 
Let $\bfphi^*=\Sel(\bfphi)$.
Recall that we use a boldface greek character, $\bfphi$, to denote the distribution of a population in $\ZZ^\ell$, and the lightface version of that character, $\phi_i$ to denote the distribution  at the $i$th locus.
We denote the mean fitness at locus $i$ by $W_i=E(\phi_i)$.  By the linearity of expectation we have $M=\sum_{i=1}^\ell W_i$.
We use $\Wha_i$ to denote the mean fitness of the loci other than $i$, i.e., $\Wha_i=M-W_i$.
Use use $V_i$ to denote the variance in fitness at the $i$th locus: $V_i=\Var(\phi_i)$.
The notation is analogous for $\bfphi^*$:  $W^*_i=E(\phi^*_i)$, $V^*_i=\Var(\phi^*_i)$, etc.

\begin{lemma}\label{lemma: selection 1 locus}
If selection acts on a population at linkage equilibrium, the effect on the $i$th locus is given by:
\[
\phi^*_i(x) = \frac{1}{M}\left(x+ \Wha_i\right)\phi_i(x).
\]
\end{lemma}
\begin{proof}
First, let us observe that $E(F(X)\ | \ X_i=x) = x+\Wha_i$:  
\[
E(F(X)\ | \ X_i=x) \ =\  \sum_{j=1}^{\ell}E(X_j| X_i=x)\ =\  x+ \sum_{j\neq i}E(X_j)\ =\  x+ \sum_{j\neq i}W_j.
\]
For $x\in \ZZ$ and $\bfy\in \ZZ^{\ell-1}$ let $x\  {\widehat{\ }_i} \bfy$  be the vector of length $\ell$ with $x$ as the $i$-coordinate and with all other coordinates given by $\bfy$ in corresponding order.  Second, we calculate $\phi^*_i(x)$:
\begin{eqnarray*}
\phi^*_i(x)  		&=&		\sum_{\bfy\in\ZZ^{\ell-1}} \bfphi^*(x\  {\widehat{\ }_i} \bfy)  \ \  = \ \   ( 1/M)\	\sum_{\bfy\in\ZZ^{\ell-1}} F(x\  {\widehat{\ }_i} \bfy)\bfphi(x\  {\widehat{\ }_i} \bfy) \\
			&=&	(1/M)\	\phi_i(x)\ E(F(X)\ | \ X_i=x). 
\end{eqnarray*}
Putting these equations together, we get the result of the lemma.
\end{proof}

The next lemma shows the effect of selection on the fitness and variance at a single locus.

\begin{lemma}
If selection acts on a population at linkage equilibrium, the effect  on fitness and variance at locus $i$ is given by:
\begin{eqnarray*}
W^*_i-W_i 		&=&		 \frac{V_i}{M},   		\\
V_i^*-V_i			&=&		 \frac{\kappa_{3,i}}{M}-\left(\frac{V_i}{M} \right)^2.			\\
\end{eqnarray*}
\end{lemma}
\begin{proof}
For the first equation:
\begin{eqnarray*}
W^*_i   	&=&		\sum_xx\phi^*_i(x)		\\
			&=&	(1/M)\		\sum_xx(x+\Wha_i)\phi_i(x) 	\\
			&=&	(1/M)\   	\left(\sum_x x^2\phi_i(x) 	+\Wha_i\sum_xx\phi_i(x)\right)\\
			&=&	(1/M)\   	\left((V_i+W_i^2)  +\Wha_i W_i\right)    \\
			&=&	(1/M)\   	(V_i+M W_i). 
\end{eqnarray*}
This establishes the first equation of the lemma.

For the second equation, let $\Vati_i$ be the {\em  second moment about the origin of $\phi_i$}, that is, $\Vati_i=\sum_{x}x^2\phi_i(x)$, and analogously for $\phi^*_i$.
Let $\kati_{3,i}$ be the {\em third moment about the origin of $\phi_i$}, that is, $\kati_{3,i}=\sum_{x}x^3\phi_i(x)$.
Then
\begin{eqnarray*}
\Vati^*   	&=&		\sum_xx^2\phi^*_i(x)		\\
			&=&	(1/M)\		\sum_xx^2\left(x+\Wha_i\right)\phi_i(x)			\\
			&=&(1/M)\   	\left(\sum_x x^3\phi_i(x) 	+\Wha_i\sum_xx^2\phi_i(x)\right)\\
			&=&	(1/M)\  	\left(\kati_{3,i}  +\Wha_i\Vati_i\right).
\end{eqnarray*}
Now, using the developments of the moments about the origin in terms of the central moments we get: 
\begin{eqnarray*}
\Vati^*_i		&=&	 V^*_i+{W^*_i}^2 = V^*_i+W_i^2 + 2V_iW_i/M+ (V_i/M)^2,\\
\kati_{3,i}  &=&  \kappa_{3,i}+3 V_i W_i+ W_i^3,\\
\Vati_i		&=& V_i+W_i^2.
\end{eqnarray*}
The equation above then becomes
\[
V^*_i+W_i^2 + 2V_i \frac{W_i}{M}+ \left(\frac{V_i}{M}\right)^2  = \frac{\kappa_{3,i}}{M}+3 V_i \frac{W_i}{M}+ \frac{W_i^3}{M} + \frac{\Wha_i}{M}(V_i+W_i^2),
\]
which we can re-arrange as
\[
V^*_i + \left(2V_i \frac{W_i}{M} - 3 V_i \frac{W_i}{M} - V_i\frac{\Wha_i}{M}\right) = \frac{\kappa_{3,i}}{M} - \left(\frac{V_i}{M}\right)^2  + \left(\frac{W_i^3}{M}+\frac{\Wha_i}{M}W_i^2 - W_i^2\right).
\]
To finish the proof of the lemma one only has to observe that $\left(2V_i \frac{W_i}{M} - 3 V_i \frac{W_i}{M} - V_i\frac{\Wha_i}{M}\right)=-V_i$ and that $\left(\frac{W_i^3}{M}+\frac{\Wha_i}{M}W_i^2 - W_i^2\right)=0$.
\end{proof}

We now continue with the proof of Theorem \ref{thm: effect rec}.
Using that $LD_2 = V-\sum_i V_i$, we get: 

\begin{eqnarray*}
LD_2(\bfphi^*)-LD_2(\bfphi) 	&=& (V^*-V) - (\sum_i V^*_i-V_i) \\
						&=&	 \left(\frac{\kappa_3}{M}-\left(\frac{V}{M} \right)^2\right) - \sum_i \left(\frac{\kappa_{3,i}}{M}-\left(\frac{V_i}{M} \right)^2\right) \\
						&=& \frac{LD_3}{M} + \frac{(\sum_{i} V_i^2) - V^2 }{M^2} \\ 
						&=&  -\frac{\sum_{i\neq j} V_iV_j}{M^2}
\end{eqnarray*}

\noindent The last equality follows since $LD_3=0$ for a population at linkage equilibrium.

\subsection{The ordering on distributions}   \label{ss: theorem 2}

This subsection is dedicated to proving some basic combinatorial lemmas  which are required for the proof of Theorem \ref{thm: rec positive}.
Our new key notion is the ordering $\preceq$ among probability distributions on $\ZZ$, which will be useful throughout the rest of the paper. We made the assumption earlier that all cumulants of  populations are finite. It is similarly to be assumed that all cumulants of distributions discussed in this section are finite. 

\begin{definition} \label{preqdef}
Given two distributions $\psi_1$ and $\psi_2\colon\ZZ\to\RR^{\geq 0}$, we define: 
\[
\psi_2\preceq \psi_1
\quad\iff\quad
(\forall b_1<b_2\in \ZZ)\ \ 
\psi_1(b_1)\psi_2(b_2) \leq \psi_1(b_2)\psi_2(b_1).
\]
We let $\psi_2\prec \psi_1$ if, in addition,  there exist $b_1<b_2\in \ZZ$ with $\psi_1(b_1)\psi_2(b_2) < \psi_1(b_2)\psi_2(b_1)$. 
\end{definition}

To give some intuition for the meaning of  $\preceq$, let us remark that if $\psi_1$ and $\psi_2$ are non-zero on an interval $[A,B]$, and zero elsewhere, then:   
\[
\psi_2\preceq \psi_1
\quad\iff\quad
(\forall b\in\ZZ \mbox{ with }A\leq b<B)\ \ 
\frac{\psi_2(b+1)}{\psi_2(b)} \leq \frac{\psi_1(b+1)}{\psi_1(b)}.
\]

If $\psi_2 \preceq \psi_1$, this gives a lot of information about the supports of $\psi_1$ and $\psi_2$ (i.e.\ those $x$ for which $\psi_1(x)\neq 0$ or $\psi_2(x)\neq 0$). If $x$ is in the support of $\psi_1$, then for any $y>x$ in the support of $\psi_2$, $y$ must also be in the support of $\psi_1$. Similarly, if $x$ is in the support of $\psi_2$, then for any $y<x$ in the support of $\psi_1$, $y$ must also be in the support of $\psi_2$. We can therefore find disjoint (possibly empty) sets $\Pi_1,\Pi_2$ and $\Pi_3$ such that the support of $\psi_2$ is $\Pi_1\cup\Pi_2$, the support of $\psi_1$ is $\Pi_2\cup \Pi_3$, and all the elements of $\Pi_1$ are below all the elements of $\Pi_2$ which are all below all the elements of $\Pi_3$.


The main three properties of the ordering $\preceq$ are that it is preserved by mutation, it is preserved by selection, and it preserves the ordering of expected values.
The proof of Theorem \ref{thm: rec positive} in the next section will use all of these lemmas to show that $LD_2$ becomes and remains negative during an asex process initially at linkage equilibrium.

\begin{lemma} \label{mutpres} 
The orderings $\prec$ and $\preceq$ are preserved by mutation.
That is: 
\[
\psi_2\preceq \psi_1
\quad\implies\quad
\Mut(\psi_2)\preceq  \Mut(\psi_1).
\]
The same holds for $\prec$.
Here $\Mut$ refers to the mutation operation for $\ell=1$.
\end{lemma}
\begin{proof}
We must show that for any values $b_2>b_1$: 

\begin{equation}
\sum_{d} \psi_2(d)\mu(b_2-d) \cdot \sum_{c} \psi_1(c)\mu(b_1-c)  \leq
\sum_{d}  \psi_1(d)\mu(b_2-d)  \cdot \sum_{c} \psi_2(c)\mu(b_1-c). 
\label{desire}
\end{equation} 
The r.h.s.\ can be re-expressed: 

\begin{align*} 
 \sum_{c} \Big(  \psi_1(c) \psi_2(c)\mu(b_2-c)\mu(b_1-c))  +  \sum_{d>c} \big((\psi_1(d) \psi_2(c)\mu(b_2-d)\mu(b_1-c) \\
+  \psi_1(c) \psi_2(d)\mu(b_2-c)\mu(b_1-d)\big) \Big) . 
\end{align*} 

\noindent The l.h.s.\ is: 
\begin{align*} 
 \sum_{c} \Big(  \psi_2(c) \psi_1(c)\mu(b_2-c)\mu(b_1-c)  +  \sum_{d>c} \big((\psi_2(d) \psi_1(c)\mu(b_2-d)\mu(b_1-c) \\
+  \psi_2(c) \psi_1(d)\mu(b_2-c)\mu(b_1-d)\big) \Big) . 
\end{align*}

\noindent For any given pair $(d,c)$ such that $d>c$ define: 
\[ \alpha_1=   \psi_1(c) \psi_2(d),\ \alpha_2= \psi_1(d) \psi_2(c), \] 
\[ \beta_1= \mu(b_2-c)\mu(b_1-d),\ \beta_2= \mu(b_2-d)\mu(b_1-c). \]

\noindent Now for any values $d>c$ we have $\alpha_2\geq \alpha_1$ because $\psi_2\preceq \psi_1$. 
We claim that we also have $\beta_2\geq\beta_1$: this holds because in order to have $\beta_1>0$ one requires $b_2\leq c+1$ and $d\leq b_1+1$, which can only be the case if $b_1+1=c+1=b_2=d$.
In that case $\beta_1=\mu(1)\mu(-1)$ and $\beta_2=\mu(0)\mu(0)$, and it follows that $\beta_2>\beta_1$ from our assumption that $\mu(0)>\mu(-1)>\mu(1)$.
Thus: 
\[
\alpha_2\beta_2 + \alpha_1 \beta_1 \geq \alpha_1 \beta_2 +\alpha_2 \beta_1.
\]
This establishes the inequality (\ref{desire}). 

Now suppose that $\psi_2 \prec \psi_1$, and let $b_1<b_2$ be such that $\psi_1(b_1)\psi_2(b_2) < \psi_1(b_2)\psi_2(b_1)$. Consider again the expansions of the l.h.s.\ and r.h.s.\ of (\ref{desire}). Since we have already shown that each term on the r.h.s.\ is greater than or equal to the corresponding term on the left, we need only identify one term on the right which is strictly greater than the corresponding term on the left. The reasoning above already suffices to give this strict inequality for the case $c=b_1,d=b_2$, since then $\alpha_2>\alpha_1$ and $\beta_2=\mu(0)^2>\beta_1$. 
\end{proof}

The next lemma shows that $\preceq$ is also preserved by selection.
In fact we shall prove a stronger result.
For $\ell=1$ and $W\in \RR$, we define a new form of selection, which, as we saw in Lemma \ref{lemma: selection 1 locus}, allows us to understand the effect of selection on a single locus under certain conditions.
For $\phi$ a probability distribution on $\ZZ$ and $x\in \ZZ$, we define
\[
\Sel_W(\phi) (x) = \left(\frac{1}{s}\right)(x+W)\phi(x),
\]
where $s$ is the normalising factor required to make $\Sel_W(\phi)$ a probability distribution: $s=\sum_{x\in\ZZ}(x+W)\phi(x) = E(\phi)+W$. We call a probability distribution on $\ZZ$ \emph{non-trivial} if its support consists of more than one point.

\begin{lemma} \label{lemma: Sel preserves prec}
If $W_1\leq W_2$ and $\psi_2\preceq \psi_1$, then $\Sel_{W_2}(\psi_2)\preceq \Sel_{W_1}(\psi_1)$. Furthermore, 
if $W_1<W_2$, $\psi_2\preceq \psi_1$ and at least one of $\psi_1$ and $\psi_2$ is non-trivial, then $\Sel_{W_2}(\psi_2)\prec \Sel_{W_1}(\psi_1)$.
\end{lemma}
\begin{proof}
Let $\psi_1^*=\Sel_{W_1}(\psi_1)$, let $\psi_2^*=\Sel_{W_2}(\psi_2)$  and consider $b_1<b_2$.
On the one side we have
\[
\psi_1^*(b_1)\psi_2^*(b_2) = \frac{1}{s_1s_2}(b_1+W_1)(b_2+W_2)\psi_1(b_1)\psi_2(b_2),
\]
which we need to show is less than or equal to
\[
\psi_1^*(b_2)\psi_2^*(b_1) = \frac{1}{s_1s_2}(b_2+W_1)(b_1+W_2)\psi_1(b_2)\psi_2(b_1),
\]
where $s_1$ and $s_2$ are the normalising factors for $\psi_1$ and $\psi_2$.
We know that $\psi_1(b_1)\psi_2(b_2) \leq \psi_1(b_2)\psi_2(b_1)$, so it is enough to show that $(b_1+W_1)(b_2+W_2)\leq (b_2+W_1)(b_1+W_2)$.
For this, one just needs to observe that:
\[
(b_2+W_1)(b_1+W_2) - (b_1+W_1)(b_2+W_2)  = (b_2-b_1)(W_2-W_1)\geq 0.
\]
\noindent Note that this actually suffices to show $(b_1+W_1)(b_2+W_2)< (b_2+W_1)(b_1+W_2)$ if $W_1<W_2$.

Now suppose that we also have $W_1<W_2$. The reasoning above actually suffices to show \emph{for all pairs} $b_1<b_2$ that $\psi_1^*(b_1)\psi_2^*(b_2) < \psi_1^*(b_2)\psi_2^*(b_1)$, \emph{so long as} $\psi_1(b_2)\psi_2(b_1)>0$. If at least one of $\psi_1$ and $\psi_2$ is non-trivial then there exists a pair $b_1<b_2$ with $\psi_1(b_2)\psi_2(b_1)>0$, giving $\Sel_{W_2}(\psi_2)\prec \Sel_{W_1}(\psi_1)$ as required.
\end{proof}

\begin{lemma}\label{lemma prec E}
If $\psi_2 \preceq \psi_1$, then $E(\psi_2)\leq E(\psi_1)$. Furthermore, if $\psi_2 \prec \psi_1$ then $E(\psi_2)< E(\psi_1)$.
\end{lemma}
\begin{proof}
The proof is divided into various cases depending on the supports of $\psi_1$ and $\psi_2$. Let $\Pi_1,\Pi_2$ and $\Pi_3$ be as defined subsequent to Definition \ref{preqdef}. 

Case 1: The support of both $\psi_1$ and $\psi_2$ is a finite interval $[A,B]$ (so $\Pi_1=\Pi_3=\emptyset$ and $\Pi_2=[A,B]$).
This is the simplest case, but gives the principal idea for the entire proof. We will define probability density functions $\varphi_i$ for $i\in [A,B]$, with $\psi_2=\varphi_A \preceq \varphi_{A+1}\preceq\cdots\preceq \varphi_B=\psi_1$, and $E(\varphi_i)\leq E(\varphi_{i+1})$ for all $i\in [A,B)$.
Each $\varphi_i$ will satisfy:
\[
(\forall b\in[A,i))\ \ 
\frac{\varphi_i(b+1)}{\varphi_i(b)} = \frac{\psi_1(b+1)}{\psi_1(b)}
\quad\mbox{ and }\quad
(\forall b\in[i,B))\ \ 
\frac{\varphi_i(b+1)}{\varphi_i(b)} = \frac{\psi_2(b+1)}{\psi_2(b)}.
\]
Suppose we have already defined $\varphi_i$ and we want to define $\varphi_{i+1}$.
We need to change the value of $\frac{\varphi_i(i+1)}{\varphi_i(i)}$ from $\frac{\psi_2(i+1)}{\psi_2(i)}$ to $\frac{\psi_1(i+1)}{\psi_1(i)}$ without changing any of the other fractions.
For this, we need to find values $c,d$ such that defining $\varphi_{i+1}(b) = c\ \varphi_i(b)$ for  $b\leq i$ and $\varphi_{i+1}(b) = d\ \varphi_i(b)$ for $b> i$ gives the required probability density function.
To find such $c$ and $d$ all one needs to do is to solve the following equation:
\begin{eqnarray*}
c S+ d (1-S) 	&=&	 	1\\
d \ \psi_1(i)\psi_2(i+1) &=& c\  \psi_1(i+1)\psi_2(i),
\end{eqnarray*} 
where $S=\sum_{j=A}^i\varphi_i(j)$.
Since $\frac{\psi_2(i+1)}{\psi_2(i)}\leq \frac{\psi_1(i+1)}{\psi_1(i)}$, we  have $c\leq 1\leq d$, and if $\frac{\psi_2(i+1)}{\psi_2(i)}< \frac{\psi_1(i+1)}{\psi_1(i)}$ then $c<1<d$.
Since we are increasing the values of $\varphi_i(b)$ for  $b>i$ and decreasing them for $b\leq i$, it is not hard to see that $E(\varphi_i)\leq E(\varphi_{i+1})$, and that if $\frac{\psi_2(i+1)}{\psi_2(i)}< \frac{\psi_1(i+1)}{\psi_1(i)}$, then $E(\varphi_i)< E(\varphi_{i+1})$.
This finishes the construction of the $\varphi_i$s and the proof for the case where the support of $\psi_1$ and $\psi_2$ is $[A,B]$.

Case 2: The support of $\psi_1$ and $\psi_2$ is not an interval, but it still finite and equal for both functions.
The proof above works almost the same way, except that one has to skip the values not in the support.

Case 3: The supports of $\psi_1$ and $\psi_2$ are equal, but while $\Pi_2$ is bounded below it is not bounded above. 
One runs the same proof, but now constructs an infinite sequence  $\psi_2=\varphi_A \preceq \varphi_{A+1}\preceq\cdots$. Ultimately $\psi_1$ is the limit of this sequence, i.e.\ for all $b$, $\psi_1(b)=\mbox{lim}_i \varphi_i(b)$. Since we have assumed that $\psi_1$ and $\psi_2$ have finite means, it follows that $E(\psi_1)=\mbox{lim}_{i\rightarrow \infty}  E(\varphi_i)$. 

Case 4: The supports of $\psi_1$ and $\psi_2$ are equal, but while $\Pi_2$ is bounded above it is not bounded below. 
One runs the same proof, but now constructs an infinite sequence  $\psi_2=\varphi_B \preceq \varphi_{B-1}\preceq\cdots$, such that  $E(\varphi_i)\leq E(\varphi_{i-1})$ for all $i\leq B$. Again we have $\psi_1$ as the limit of this sequence and  $E(\psi_1)=\mbox{lim}_{i\rightarrow \infty}  E(\varphi_i)$.

Case 5: The supports of $\psi_1$ and $\psi_2$ are equal, and  $\Pi_2$ neither bounded above nor bounded below. 
One runs almost the same proof, but now in two stages. Choosing $A\in \Pi_2$, we first construct an infinite sequence  $\psi_2=\varphi_A \preceq \varphi_{A+1}\preceq\cdots$ which has the intermediate value $\psi_3$ as limit. One then constructs an infinite sequence $\psi_3=\varphi_A' \preceq \varphi_{A-1}'\preceq\cdots$ with $\psi_1$ as limit. 

Case 6: At least one of $\Pi_1$ or $\Pi_3$ is non-empty. If $\Pi_2$ is empty then it immediately follows that $E(\psi_2)<E(\psi_1)$, so suppose this does not hold. Let $\psi_1^{\ast}$ be the probability density function formed from $\psi_1$ by restricting the support to $\Pi_2$ (and normalising as appropriate), and form $\psi_2^{\ast}$ similarly. If $\Pi_1$ is non-empty then we have: 

\[ E(\psi_2)<E(\psi_2^{\ast})\leq E(\psi_1^{\ast})\leq E(\psi_1). \] 
\noindent If $\Pi_3$ is non-empty then we have:

\[ E(\psi_2)\leq E(\psi_2^{\ast})\leq E(\psi_1^{\ast}) <  E(\psi_1). \]  
\end{proof}

\subsection{The properties $(\dagger)$ and $(\dagger \dagger)$}   \label{ss: theorem 2}

This section is dedicated to proving  Theorem \ref{thm: rec positive}, which asserts that $LD_2$ stays negative throughout the process, independent of what operations are applied and in which order, except when recombination has just been applied in which case $LD_2=0$.
This is for the $\ZZ$-model, and assuming truncation is applied after each application of mutation and recombination (or at least before any application of selection).

We previously described a distribution on $\ZZ$ as non-trivial if there exists more than more point in the support. We shall refer to a population $\bfphi$ (on $\ZZ^{\ell}$ for $\ell >1$) as non-trivial if at least two of the marginal distributions $\phi_i$ are non-trivial (and providing this remains the case when truncation is applied to $\bfphi$). 

The key idea behind the proof of Theorem \ref{thm: rec positive} is to consider a property, $(\dagger)$, which suffices to ensure that $LD_2$ is non-positive. We will also consider a strengthening of $(\dagger)$, which we call $(\dagger\dagger)$ and which ensures that $LD_2$ is negative.
To prove Theorem \ref{thm: rec positive} we then use induction on the generations and show that the property $(\dagger)$ is satisfied at each step of the process, and that, in fact, the stronger property $(\dagger\dagger)$ is also satisfied from some early stage onwards (the first stage after which $\Sel$ is applied to a non-trivial population).

\begin{definition} \label{dagstuff} Given a probability distribution $\psi\colon \ZZ^2\to\RR^{\geq 0}$ and $a\in \ZZ$, we shall say that $\psi^a$ is defined if $\psi(a)=\sum_{b} \psi (a,b) \neq 0$. In this case $\psi^a$ is the distribution given by $\psi^a(b)=\psi(a,b)/\psi(a)$. 

We say $\psi$ satisfies $(\dagger)$ if for every $a_1<a_2$ such that $\psi^{a_1}$ and $\psi^{a_2}$ are defined,  $\psi^{a_2}\preceq \psi^{a_1}$. 
We say that $\psi$ satisfies $(\dagger\dagger)$ if, in addition, there exist $a_1<a_2$ such that $\psi^{a_1}$ and $\psi^{a_2}$ are defined and $\psi^{a_2}\prec \psi^{a_1}$ as witnessed by a pair $b_1<b_2$ with  $a_1+b_1>0$.
\end{definition}

To define $(\dagger)$ for a population with $\ell>2$, we need to consider each locus compared to the rest of the loci altogether. 
For  $i\neq j\in\{1,...,\ell\}$, let $F_i(X)= X_1+\cdots + X_{i-1}+X_{i+1}+\cdots+\X_\ell = F(X)-X_i$.
Given a population $\bfphi$, let $\phiha_{i}$ be the distribution corresponding to the random variable $(X_i,F_i(X))$.
Equivalently: 
\[
\phiha_{i}(a,b)=\sum_{\substack{\bfx\in \ZZ^\ell,\\ x_i=a,\ F(\bfx)=a+b}}\bfphi(\bfx).
\]

\begin{definition}
A population $\bfphi$ satisfies $(\dagger)$ if $\phiha_i$ does for every $i=1,...,\ell$. 
A population $\bfphi$ satisfies $(\dagger\dagger)$ if there exists $i$ such that $\phiha_i$ satisfies $(\dagger\dagger)$.
\end{definition}

The next step is to prove that $(\dagger)$ and $(\dagger\dagger)$ are preserved through the operations.
Recall that we are assuming the process starts at linkage equilibrium. 
Note that if $\bfphi$ is at linkage equilibrium, then $(\dagger)$ holds -- in that case we have equality between the left-hand side and the right-hand side in the definition of the $\preceq$ relation.

\begin{lemma}\label{lemma dagger sel}
For $\ell\geq 2$, if $\bfphi$ satisfies $(\dagger)$, then $\Sel(\bfphi)$ satisfies $(\dagger)$, and if $\bfphi$ is non-trivial then $\Sel(\bfphi)$ satisfies $(\dagger\dagger)$.
\end{lemma}
\begin{proof}
Let $\bfphi^*=\Sel(\bfphi)$.
Fix $i\in\{1,...,\ell\}$.
A similar argument to that of Lemma \ref{lemma: selection 1 locus} shows that $\phiha_i^*(a,b)=(1/M)\left(a+b\right)\phiha_i(a,b)$.
Thus, $(\phiha_i^*)^a = \Sel_a((\phiha_i)^{a}).$
Suppose $a_1<a_2$ are such that $(\phiha_i^*)^{a_1}$ and $(\phiha_i^*)^{a_2}$ are both defined.
It follows from Lemma \ref{lemma: Sel preserves prec} that, since  $(\phiha_i)^{a_2}\preceq (\phiha_i)^{a_1}$, we have $(\phiha_i^*)^{a_2}\preceq (\phiha_i^*)^{a_1}$. If $\bfphi$ is non-trivial it also follows from Lemma \ref{lemma: Sel preserves prec} that $(\phiha_i^*)^{a_2}\prec (\phiha_i^*)^{a_1}$. 
Thus $\phiha_i^*$ satisfies $(\dagger)$ and also satisfies $(\dagger\dagger)$ if $\bfphi$ is non-trivial, as required.
\end{proof}

\begin{lemma}\label{lemma dagger mut}
Both $(\dagger)$ and $(\dagger\dagger)$ are preserved by mutation.
\end{lemma}
\begin{proof}
First, let us note that mutation can be broken down into a number of consecutive steps, by treating one locus at a time.
Let $\Mut^k$ be the application of mutation on the $k$th locus, i.e., $\Mut^k(X_1,...,X_k,...,X_\ell)=(X_1,...,X_k+C_k,...,X_\ell)$.
We will show that both  $(\dagger)$ and $(\dagger\dagger)$ are preserved by each of the operations $\Mut^k$.
Let $\bfphi^*=\Mut^k(\bfphi)$.
Fix $i$.
Assuming $(\dagger)$ or $(\dagger\dagger)$ for $\phiha_i$, we establish that the same condition holds  for $\phiha_i^*$.

Suppose first that $k\neq i$.
Then for $a\in \ZZ$, $(\phiha_i^*)^{a} = \Mut((\phiha_i)^a)$, because the mutation happens in one of the loci included in the second coordinate, and has the same effect on $F_i(X)$.
Consider $a_1<a_2$ such that $(\phiha_i)^{a_1}$ and $(\phiha_i)^{a_2}$ are both defined.
It then follows from Lemma \ref{mutpres} that since $(\phiha_i)^{a_2}\preceq (\phiha_i)^{a_1}$ we have $(\phiha_i^*)^{a_2}\preceq (\phiha_i^*)^{a_1}$, and hence that $\phiha_i^*$ satisfies $(\dagger)$.
We get $(\dagger\dagger)$ similarly.  

If $k=i$, then the proof is the same.
One just needs to observe that $\psi\colon \ZZ^2\to\RR^{\geq 0}$ satisfies $(\dagger)$ (or $(\dagger\dagger)$) if and only if $\psi'(a,b)=\psi(b,a)$ does.
\end{proof}

So far both lemmas hold for any of the models.
The following lemma only holds for the $\ZZ$-model.

\begin{lemma}\label{lemma dagger trunc}
Both $(\dagger)$ and $(\dagger\dagger)$ are preserved by truncation for the $\ZZ$-model.
\end{lemma}
\begin{proof}
Let $\bfphi^*$ be the population which results from an application of truncation to $\bfphi$. 
Let $s=\sum_{\bfx\in D}\bfphi(\bfx)$, where $D$ is as in the $\ZZ$-model.
Note that $\phiha_i^*(a,b)=0$ if $a+b\leq 0$ and $\phiha_i^*(a,b)=\phiha_i(a,b)/s$ if $a+b>0$.

Fix $i$, $a_1<a_2$ and $b_1<b_2$.
Then if $\phiha_i^*(a_1,b_1)\neq 0$ we have $a_1+b_1>0$,  and hence both $a_1+b_2$ and $a_2+b_1$ are positive. 
Therefore, if $(\phiha_i)^{a_1}(b_2)(\phiha_i)^{a_2}(b_1)\geq (\phiha_i)^{a_2}(b_2)(\phiha_i)^{a_1}(b_1)$ then  $(\phiha_i^*)^{a_1}(b_2)(\phiha_i^*)^{a_2}(b_1)\geq (\phiha_i^*)^{a_2}(b_2)(\phiha_i^*)^{a_1}(b_1)$. In the same way 
if $(\dagger\dagger)$ holds because $(\phiha_i)^{a_1}(b_2)(\phiha_i)^{a_2}(b_1)> (\phiha_i)^{a_2}(b_2)(\phiha_i)^{a_1}(b_1)$ holds and $a_1+b_1>0$ (the latter condition being required by Definition \ref{dagstuff}), then this implies $(\phiha_i^*)^{a_1}(b_2)(\phiha_i^*)^{a_2}(b_1)> (\phiha_i^*)^{a_2}(b_2)(\phiha_i^*)^{a_1}(b_1)$.
Thus both $(\dagger)$ and $(\dagger\dagger)$ are preserved, as required.
\end{proof}

The second part of the proof of Theorem \ref{thm: rec positive} is to show the connection between the  properties $(\dagger)$, $(\dagger\dagger)$ and $LD_2$.
Recall that $\Co(X,Y)$ is the {\em covariance} of the random variables $X$ and $Y$, i.e., $\Co(X,Y)=E(XY)-E(X)E(Y)$.
Using our calculations from \S\ref{sss: LD2} we obtain that: 
\[
LD_2(\bfphi) =  \sum_{i=1}^\ell Co(\phiha_i).
\]

\begin{lemma}
If $\varphi\colon\ZZ^2\to\RR^{\geq 0}$ satisfies $(\dagger)$, then $\Co(\varphi)\leq 0$.
Furthermore, if $\varphi$ satisfies $(\dagger\dagger)$ then $\Co(\varphi)<0$.
\end{lemma}
\begin{proof}
Let $X,Y$ be random variables such that $(X,Y)$ has probability distribution $\varphi$.
The key idea is to use that $(\dagger)$ implies that $E(Y|X=a)$ is decreasing in $a$, which follows from Lemma \ref{lemma prec E}.

Set $v=E(X)$. Now we would like to put $u=E(Y|X=v)$, but since we may have $v\notin \ZZ$ this presents a slight difficulty.
If $v\notin\ZZ$, we let $u$ be a number in between $E(Y|X=\lfloor v\rfloor)$ and $E(Y|X=\lceil v\rceil)$.
Now let $W=X-v$  and $ Z=Y-u$.  Since $v$ and $u$ are constants we have: 

\[ \Co (X,Y)=\Co (W,Z)=E(WZ)-E(W)E(Z). \] 

Let $\psi(W,Z)$ specify the distribution on the pair $(W,Z)$. 
Notice that since $\varphi$ satisfies $(\dagger)$, so does $\psi$.
 Let $\psi_W$ and $\psi_Z$ specify  the corresponding marginal distributions, and let $\psi_Z(Z|W)$ specify the conditional distribution.  
We claim that  $E(WZ)$ is non-positive. This is because: 
\[  \sum_{a} \sum_{b} \psi(a,b)\cdot ab \ =\  \sum_{a} a\psi_W(a)  \sum_{b} b \psi_Z(b|a) \ = \  \sum_{a} a \psi_W(a) E(Z|W=a). \] 

Now satisfaction of $(\dagger)$ and Lemma \ref{lemma prec E} imply that when $a$ is positive $E(Z|W=a)$ is non-positive, and when $a$ is negative $E(Z|W=a)$ is non-negative. 
Also $E(W)E(Z)=0$  because $E(W)=0$. 

If $\varphi$ satisfies $(\dagger\dagger)$ then $E(Z|W=a)$ is non-zero somewhere, and hence $E(WZ)<0$, as needed to get $\Co (X,Y)<0$.
\end{proof}

\begin{corollary}\label{cor: LD 2 negative}
If $\bfphi$ satisfies $(\dagger)$, then $LD_2(\bfphi)\leq 0$.
Furthermore, if $\bfphi$ satisfies $(\dagger\dagger)$ then $LD_2(\bfphi)<0$.
\end{corollary}

To finish the proof of Theorem \ref{thm: rec positive}, all we need to observe is that since the population starts at linkage equilibrium, it initially satisfies $(\dagger)$.
By Lemmas \ref{lemma dagger sel}, \ref{lemma dagger mut} and \ref{lemma dagger trunc}, the condition $(\dagger)$ is then satisfied throughout the process. Furthermore, the condition $(\dagger\dagger)$ is satisfied after an application of $\Sel$ to any non-trivial population (mutation will quickly produce non-trivial populations), and then remains satisfied until any such point as $\Rec$ is applied. Since satisfaction of $(\dagger\dagger)$ ensures negative $LD_2$, any non-trivial application of $\Rec$ therefore increases variance.


\subsection{The bounded model}  \label{ss: theorem 3}

This section provides a deep analysis of the asymptotic behavior of the bounded model and is dedicated to giving  the full proof of Theorem \ref{main bounded alleles}.
Recall that in the bounded models the fitness values of the genes are restricted to $\{1,...,N\}$, which we denote $[1,N]$, and the domain of the process is $D=[1,N]^\ell$, which we sometimes denote $N^{\ell}$.
Also recall our assumption that our mutation distributions only take non-zero values on $-1$, $0$ and $1$; we let $a=\mu(-1)$, $b=\mu(0)$ and $c=\mu(1)$.
Even though that assumption seems to be inessential in simulations, it is important for the type of formal analysis we do in this section. 
Recall also that we assume that $c<a<b$.

Let us restate Theorem \ref{main bounded alleles}.
In the statement of the theorem, we consider the case where both types of individuals, sexual and asexual, live together.
Mutation and selection act the same way on both, but recombination acts only among the sexual  individuals.
Imagine, for instance, that at a certain time there exists a purely asexual population, and that by random chance a few of the individuals in the population then become sexual and start reproducing among themselves.
The theorem states that over time the proportion of the population which is sexual will then tend to 1. 

\vspace{0.2cm}
\noindent \textbf{Theorem 3}. 
\emph{For every $a,b,c$ with $c<a$,  for every $\ell>1$, and for all sufficiently large $N$ (i.e.\ there exists $N_0$ such that $\forall N\geq N_0$), whatever the initial population is, so long as the proportion of sexual individuals is non-zero, we have that the proportion of sexual individuals converges to 1 and the proportion of asexual ones converges to 0.}\\

Let us fix the values of $a,b,c$ and $\ell$ throughout the rest of Section \ref{ss: theorem 3}.

The proof of Theorem \ref{main bounded alleles} proceeds by showing that the mean fitness of the sexual population is eventually higher than that of the asexual population, over a time average.
We will be able to show that the asexual population in isolation converges to a limit distribution, and we will provide an upper bound for the mean fitness of that limit.
While we strongly suspect that the sexual population in isolation also converges to a limit (as evidenced by simulations), we have not been able to prove it.
Nevertheless, we can still provide a lower bound for the geometric mean of the mean fitness over generations, which is larger than the upper bound for asex. We will show this is enough to establish that sex outperforms asex.
To obtain these upper and lower bounds, the key technique is to study the case when positive and negative mutations have the same probabilities. This case is much easier to analyse, and we then find a way of translating those results to the case we are interested in, where downward mutations are more likely. 
An issue that we have to be constantly aware of is how truncation affects the populations.

Let us start by showing that we can analyse the sexual and asexual populations separately, as we were doing earlier in the paper.
Let $\bfphi^t\in \RR^{D\times \{\mathtt{s,a}\}}$ be the probability distribution for the entire population at stage $t$ (no confusion should result from any conflict in notation with that given in Definition \ref{dagstuff} -- the latter notation may be considered only to apply to the proof of Theorem \ref{thm: rec positive}).
We can split $\bfphi^t$ into two functions $\bfphi_\sextt^t$ and $\bfphi_\asextt^t$, each with domain $D$.
Let $\varphi_\sextt^t$ and $\varphi_\asextt^t$ be the proportions of individuals which are sexual and asexual respectively, i.e., $\varphi_\sextt^t=|\bfphi_\sextt^t|$ and $\varphi_\asextt^t=|\bfphi_\asextt^t|$, where $|g|$ is the {\em taxicab norm}: For $\bfpsi\in \RR^D$,
\[
|\bfpsi|=\sum_{\bfx\in D}|\bfpsi(\bfx)|.
\]

\noindent So far, we have always assumed that populations, $\bfpsi$, are probability distributions, and hence that $|\bfpsi|=1$.
This is not the case with $\bfphi_\sextt^t$ and $\bfphi_\asextt^t$, and we will not be making that assumption anymore.
If a probability distribution is really required, all we have to do is consider normalisation:
\[
\Proje(\bfpsi)= \frac{\bfpsi}{|\bfpsi|}.
\]
From now on when we refer to the operations of truncation and selection, we will omit the normalisation:
Thus, when we apply truncation, giving $\Truncation(\bfpsi)$, we just erase the individuals outside the domain, and applying selection, giving $\Sel(\bfpsi)$, simply involves multiplying $\bfpsi(\bfx)$ by $F(\bfx)$.
We do not alter the definition of $\Mut$ as this operation already preserves the norm.  The definition of $\Rec$ is now altered in the obvious manner, so that it remains norm preserving (requiring division by $| \bfpsi |^{\ell-1}$). 
Since the normalization operation, $\Proje(\bfpsi)$, commutes with all the other operations, it does not really matter when we apply it.
The advantage of describing the same process in this new fashion, is that now mutation, truncation and selection act independently for the sexual and asexual populations,  and it is only normalisation that involves interaction between them. 
Of the operations selection, recombination, mutation and truncation, only truncation and selection affect the values  $\varphi_\sextt$ and $\varphi_\asextt$.
Since recombination and mutation commute, we can and will assume in what follows that  the operations cycle through in that order. In fact, we will assume that $\bfphi^{t+1}$ is obtained from $\bfphi^t$ by applying: mutation, truncation, selection and recombination  in that order (recombination, of course, being applied for sex only).
For each $t$, let: 
\[
\lambda_\sextt^t = \frac{|\Rec(\Sel(\Truncation(\Mut(\bfphi^t_\sextt))))|}{|\bfphi^t_\sextt|}
\quad{\and}\quad
\lambda_\asextt^t = \frac{|\Sel(\Truncation(\Mut(\bfphi^t_\asextt)))|}{|\bfphi^t_\asextt|}.
\]
For the quotients, we have:  
\[
\frac{\varphi_\sextt^t}{\varphi_\asextt^t}  
\ \ = \ \  
\frac{\varphi_\sextt^{t-1}}{\varphi_\asextt^{t-1}}\frac{\lambda_\sextt^{t-1}}{\lambda_\asextt^{t-1}} 
\ \ = \ \
 \frac{\varphi_\sextt^0}{\varphi_\asextt^0}  \frac{ \prod_{i=0}^{t-1}\lambda_\sextt^i} {\prod_{i=0}^{t-1}\lambda_\asextt^i}. 
\]

To establish Theorem \ref{main bounded alleles}, it will be enough to show that for all sufficiently large $N$: 
\[
\lim_t\frac{\prod_{i=0}^{t-1} \lambda_\sextt^i} {\prod_{i=0}^{t-1} \lambda_\asextt^i}  = +\infty.
\]
We may therefore consider the populations separately, because all we need be concerned with are the values 
 $\prod_{i=0}^{t-1} \lambda_\sextt^i$ and $\prod_{i=0}^{t-1} \lambda_\asextt^i$, which evolve independently.
We will establish the limit above,  by showing that the geometric mean of $\lambda_\sextt^i$ is eventually always greater than that of $\lambda_\asextt^i$ by at least a fixed margin.

Let us fix the notation for dealing with geometric means.
Given a finite sequence $a_1,...,a_n$ of numbers, we let $\GM_{t\leq n}(\{a_t\}) = \sqrt[n]{\prod_{i=1}^na_i}$.
Given an infinite sequence $\{a_t\}_{t\in\NN}$ we let
\[
\GMsup_t(a_t)  = \limsup_n \sqrt[n]{\prod_{i=1}^na_i},
\quad \mbox{   }\quad
\GMinf_t(a_t)  = \liminf_n \sqrt[n]{\prod_{i=1}^na_i},
\]
and if both limits are the same we call this common value $\GM_t(a_t)$.

The rest of this subsection is dedicated to proving the following theorems:

\begin{theorem}\label{main thm Asex}
For every $N$, the limit of $\lambda_\asextt^t$ exists and, for $\tau = b+ 2\sqrt{ac}$:
\[
\lim_{t\to\infty}  \lambda_\asextt^t  <  N \ell \tau^\ell.
\]
\end{theorem}

\noindent Using facts observed from simulations, we are confident in claiming that in actual fact  $\lim_{N\to\infty}\lim_{t\to\infty}  \lambda_\asextt^t/N = \ell \tau^\ell$.
We will not need this extra fact, however,  and the result of the theorem will be enough for our purposes.

\begin{theorem}\label{main thm Sex}
Let $\tau = b+ 2\sqrt{ac}$ as above. For all sufficiently large $N$:
\[
\GMinf_{t\to\infty}(\lambda_\sextt^t)/  > N \ell \tau^\ell.
\]
\end{theorem}

\noindent Using  facts observed from simulations, we are confident in claiming  that in actual fact $\lim_{N\to\infty}\lim_{t\to\infty}  \lambda_\sextt^t/N = \ell\tau/(\ell-\tau(\ell-1))$, which is greater than $\ell\tau^\ell$ for $\ell>1$ and $\tau<1$.
Once again, we will not need this extra fact, however,  and the result of the theorem will be enough for our purposes.

It then follows from the theorems above that $\GMinf_{t\to\infty}(\lambda_\sextt^t/\lambda_\asextt^t) > 1$ for any large enough $N$, and hence that $\lim_t(\prod_{i=0}^{t-1} \lambda_\sextt^i)/(\prod_{i=0}^{t-1} \lambda_\asextt^i)  = +\infty$ as required.

\subsubsection{Understanding $\lambda^t_\sextt$ and $\lambda^t_\asextt$}

In general, given a population $\bfpsi\in \RR^{\NN^\ell}$, we define 
\[
\lambda(\bfpsi) = \frac{|(\Sel(\Truncation(\Mut(\bfpsi))))|}{|\bfpsi|}.
\]
Then $\lambda^t_\sextt=\lambda(\bfphi^t_\sextt)$,  and similarly  for asex.

Given a population $\bfpsi\in \RR^{\NN^\ell}$, let $\rho(\bfpsi)$ be the proportion of individuals  surviving mutation followed by truncation, i.e.: 
\[
\rho(\bfpsi) = \frac{|\Truncation(\Mut(\bfpsi))|}{|\bfpsi|}.
\]
Let us remark that $\rho(\bfpsi)\leq 1$.
Given a population $\bfpsi\in \RR^{\NN^\ell}$,  we use $\MF(\bfpsi)$ to denote its mean fitness, even in the case that $\bfpsi$ is not normalised:
\[
M(\bfpsi) = \frac{\sum_{\bfx\in N^\ell} F(\bfx)\bfpsi(\bfx)}{|\bfpsi|}.
\]
The increase in norm caused by an application of selection (ignoring normalization) is given by the mean fitness:
\[
\frac{|\Sel(\bfpsi)|}{|\bfpsi |} =\frac{\sum_{\bfx\in D}F(\bfx)\bfpsi(\bfx)}{|\bfpsi |} = M(\bfpsi).
\]

\noindent Let us remark that $\MF(\bfpsi)\leq N\ell$ because $F(\bfx) \leq N\ell$ for every $\bfx\in D$.

Since  mutation and recombination do not affect the norms, we have: 
\[
\lambda(\bfpsi) = \rho(\bfpsi) \MF(\bfpsi'),
\]
where $\bfpsi'=\Truncation(\Mut(\bfpsi))$.


\subsubsection{Changing the parameters}

A key idea here  is to use the case when positive and negative mutations are equiprobable to get information about the case we are interested in, where $c<a$.
In this subsection we show how we can change the mutation parameters from $a,b,c$  to $a', b', c'$ satisfying $a'=c'$,  in a manner which allows us to translate from one process to the other in a controlled way.

We define $a',b',c'$ so that they satisfy the following equations:
\[
a'+b'+c'=1,
\quad\quad
\frac{b}{\sqrt{ac}}=\frac{b'}{\sqrt{a'c'}}
\quad \mbox{and}\quad
a'=c'.
\]
The reason we require $b/\sqrt{ac}=b'/\sqrt{a'c'}$ will become clear later. These equations are enough to determine the values of $a'$, $b'$ and $c'$ as follows.
Since $a'=c'=(1-b')/2$, for $\tau=   b+ 2\sqrt{ac}$ we get:
\begin{eqnarray}
\frac{b}{\sqrt{ac}}  &=& \frac{b'}{\frac{1}{2}(1-b')}   \\
b(1-b') &=& 2b' \sqrt{ac}     \\
b & = & b'(b+2\sqrt{ac})    \\
\frac{\sqrt{ac}}{\sqrt{a'c'}}=\frac{b}{b'}&=&\tau.
\end{eqnarray}

\noindent So $b'=b/\tau$ and $a'=c' = \sqrt{ac}/\tau$.

To better visualize the translation from the $abc$-process to the $a'b'c'$-process, let us start by considering the case $\ell=1$ first.
Consider the diagonal square matrix $C$ of size $N\times N$ given by:
\[
C(x,x) = \left(a/c\right)^{x/2}.
\]
Our goal now is to show that applying the $abc$-process to a population $\bfphi$ is equivalent to applying the $a'b'c'$-process to $C\cdot \bfphi$ up to a factor of $\tau$.
In other words, we will show that $\tau\cdot \Sel(\Mut_{a'a'}(C\cdot \bfphi)) =   C\cdot  \Sel(\Mut_{ac}(\bfphi))$.

Let $\Mut_{ac}$ be the matrix corresponding to an application of mutation with probabilities $\mu(-1)=a$, $\mu(0)=b$ and $\mu(1)=c$, followed by truncation (but without normalisation).
That is:
\[
\Mut_{ac}= \left(
\begin{array}{cccccccc}
b	&	a	&	0	&	0	& 	... & 0  & 0\\
c	&	b	&	a	&	0	&	... & 0  & 0\\
0	&	c	&	b	&	a	&	...& 0  & 0 \\
0	&	0	&	c	&	b	&	\ddots &   & 0\\
\vdots & \vdots &  \vdots	& \ddots 	&   \ddots	&	\ddots &	\vdots \\  
0	&	0	&	0	&      	&    \ddots&  b	&	a\\
0	&	0	&	0	&      0	&      \dots&   c	&	b	\\
\end{array}
\right)
\]

\noindent From now on, we will assume truncation is part of mutation, and mutation refers to multiplication by  $\Mut_{ac}$.

\begin{lemma} \label{lemma: conjugating mutation}
In the case $\ell=1$:
\[
\tau \cdot \Mut_{a'a'} =   C\cdot \Mut_{ac} \cdot C^{-1}
\]
\end{lemma}
\begin{proof}
We carry out the matrix multiplications:

{\small
\[
C\cdot \Mut_{ac} \cdot C^{-1} = \left(
\begin{array}{ccccc}
b \sqrt{a/c}^{1}\sqrt{a/c}^{-1}	&	a \sqrt{a/c}^{1}\sqrt{a/c}^{-2}	&	0&	0	& 	...\\
c \sqrt{a/c}^{2}\sqrt{a/c}^{-1}	&	b	\sqrt{a/c}^{2}\sqrt{a/c}^{-2}&	a	\sqrt{a/c}^{2}\sqrt{a/c}^{-3}&	0	&	...\\
0  &	c\sqrt{a/c}^{3}\sqrt{a/c}^{-2}	&	b	\sqrt{a/c}^{3}\sqrt{a/c}^{-3}&	a	\sqrt{a/c}^{3}\sqrt{a/c}^{-4}&	...\\
0	&	0	&	c	\sqrt{a/c}^{4}\sqrt{a/c}^{-3}&	b	\sqrt{a/c}^{4}\sqrt{a/c}^{-4}&	...\\
\vdots & \vdots & 	\vdots & \vdots & 	\ddots  
\end{array}
\right)
\]}
\[
= \left(
\begin{array}{ccccc}
b 	&	a \sqrt{a/c}^{-1}	&	0	&	0	& 	...\\
c \sqrt{a/c}	&	b	&	a	\sqrt{a/c}^{-1}&	0	&	...\\
0 	&	c\sqrt{a/c}	&	b	&	a	\sqrt{a/c}^{-1}   &	...\\
0	&	0	    		&	c	\sqrt{a/c}&	b	&	...\\
\vdots & \vdots & 	\vdots & \vdots & 	\ddots  
\end{array}
\right)
= \left(
\begin{array}{ccccc}
b 	&	 \sqrt{ac}	&	0	&	0	& 	...\\
 \sqrt{ac}	&	b	&		\sqrt{ac} &	0	&	...\\
0 	&	\sqrt{ac}	&	b	&		\sqrt{ac}   &	...\\
0	&	0	    		&		\sqrt{ac}&	b	&	...\\
\vdots & \vdots & 	\vdots & \vdots & 	\ddots  
\end{array}
\right)
\]

\[
= \sqrt{\frac{ac}{a'c'}} 
\left(
\begin{array}{ccccc}
b' 	&	 \sqrt{a'c'}	&	0	&	0	& 	...\\
 \sqrt{a'c'}	&	b'	&		\sqrt{a'c'} &	0	&	...\\
0 	&	\sqrt{a'c'}	&	b'	&		\sqrt{a'c'}   &	...\\
0	&	0	    		&		\sqrt{a'c'}&	b'	&	...\\
\vdots & \vdots & 	\vdots & \vdots & 	\ddots   
\end{array}
\right)
= \tau\cdot \Mut_{a'a'}.
\]
The last equality uses that $\sqrt{a'c'}=a'=c'$ and that $ \sqrt{\frac{ac}{a'c'}} = \tau$. 
\end{proof}

Notice that, in the 1-locus case, selection without normalization is given by a diagonal matrix $\Sel$ where $\Sel(i,i)=i$.
Since diagonal matrices commute, we have:
\[
\tau\cdot \Sel\cdot \Mut_{a'a'} \cdot C =   C\cdot \Sel\cdot \Mut_{ac}.
\]

The case when $\ell>1$ is not overly different, but the notation is now a little more cumbersome.
Consider the diagonal square matrix $C$ of size $N^\ell\times N^\ell$ given by
\[
C((x_1,...,x_\ell), (x_1,...,x_\ell)) = (a_{-1}/a_{1})^{(\sum x_i)/2}
\]
where $a_{-1}, a_0, a_1$ are $a,b,c$ respectively.

Let us use $\Mut_{ac}$ to denote the matrix corresponding to $abc$-mutation with $\ell$ genes.
We should actually denote this matrix by $\Mut_{ac,\ell,N}$, but since there is no risk of confusion we prefer to simplify the notation.

\begin{lemma}\label{lemma: C commutes}
For $\ell\geq 1$:
\[
\tau^\ell \cdot\Mut_{a'a'} =   C\cdot \Mut_{ac} \cdot C^{-1}
\]
\end{lemma}
\begin{proof}
Consider $\bfpsi\in \RR^{N^\ell}$.
Then, for $\bfx=(x_1,...,x_\ell)$,
\[
C^{-1}\cdot \bfpsi(\bfx) = \bfpsi(\bfx) \sqrt{a_{-1}/a_{1}}^{-\sum x_i}
\]
and
\[
\Mut_{ac} \cdot C^{-1}\cdot \bfpsi (\bfx)  =   \sum_{i_1=-1}^1 \sum_{i_2=-1}^1 ...\sum_{i_\ell=-1}^1  \left(\prod_{j=1}^\ell  a_{i_j}\right)\bfpsi(x_1-i_i, x_2-i_2, ..., x_\ell-i_\ell) \sqrt{a_{-1}/a_{1}}^{-\sum x_j-i_j}.
\]
In the equation above, assume that if $(x_1-i_i, x_2-i_2, ..., x_\ell-i_\ell)\not\in D$, then $\bfpsi(x_1-i_i, x_2-i_2, ..., x_\ell-i_\ell)=0$.
Replacing each $a_0$ by $a_0'\tau$, each $a_{-1}$ by $a_{-1}'\tau\sqrt{a_{-1}/a_{1}}$ and each $a_1$ by  $a_1'\tau\sqrt{a_1/a_{-1}}$ we get
\[
=  \sum_{i_1=-1}^1 \sum_{i_2=-1}^1 ...\sum_{i_\ell=-1}^1 \left( \prod_{j=1}^\ell  a'_{i_j}\tau \sqrt{a_{1}/a_{-1}}^{i_j}\right) \bfpsi(x_1-i_i, x_2-i_2, ..., x_\ell-i_\ell) \sqrt{a_{-1}/a_{1}}^{-\sum_j x_j+\sum_ji_j}  
\]
\[
= \tau^\ell \sqrt{a_{-1}/a_{1}}^{-\sum_j x_j} \sum_{i_1=-1}^1 \sum_{i_2=-1}^1 ...\sum_{i_\ell=-1}^1 \left( \prod_{j=1}^\ell  a'_{i_j}  \right) \bfpsi(x_1-i_i, x_2-i_2, ..., x_\ell-i_\ell)   
\]
\[
= \tau^\ell\cdot C^{-1} \cdot\Mut_{a'a'}\cdot \bfpsi(\bfx).
\qedhere \]
\end{proof}

\subsubsection{The fixed point for Asex} \label{ss: fixedpoint asex}

In this subsection we prove Theorem \ref{main thm Asex} which states that the limit of $\lambda_\asextt^t$ is less than or equal to $N\ell\tau^\ell$.
The proof has two steps. First we show that if the asex process reaches a fixed point $\bfpsi$, then  $\lambda(\bfpsi)\leq N\ell\tau^\ell$. Second, we show that, independent of the starting point, the asex population always converges to a fixed point.

As we mentioned before, mutation (which we now consider to incorporate truncation) acts on a population (considered as a vector) by multiplying this vector by the matrix $\Mut_{ac}$.
We use $\Sel\Mut_{ac}$ to denote the matrix $\Sel\cdot\Mut_{ac}$.
In the case $\ell=1$ we have
\[
\Sel\Mut_{ac}= \left(
\begin{array}{ccccc}
b	&	a	&	0	&	0	& 	...\\
2c	&	2b	&	2a	&	0	&	...\\
0	&	3c	&	3b	&	3a	&	...\\
0	&	0	&	4c	&	4b	&	...\\
\vdots & \vdots & 	\vdots & \vdots & 	\vdots  
\end{array}
\right).
\]

\noindent The following lemma shows how useful is  the translation developed in the previous subsection.

\begin{lemma}\label{lemma: asex fixed point bound}
Suppose that $\bfpsi_{ac,\ell,N}$ is a fixed point for the asex process. Then  $\lambda(\bfpsi_{ac,\ell,N})< N\ell\tau^\ell$.
\end{lemma}
\begin{proof}
Let $\bfpsi = \bfpsi_{ac,\ell,N}$.
That $\bfpsi$ is a fixed point for the asex process means  $\bfpsi=\Proje(\Sel\Mut_{ac}\cdot \bfpsi)$, or equivalently that $\bfpsi$ is an eigenvector for $\Sel\Mut_{ac}$ with eigenvalue $\lambda(\bfpsi)$, i.e., $\Sel\Mut_{ac}\cdot\bfpsi = \lambda(\bfpsi)\bfpsi$.
From Lemma \ref{lemma: C commutes} we have that $\tau^\ell \cdot \Sel\Mut_{a'a'} \cdot C = C\cdot \Sel\Mut_{ac}$.
It follows that $\bfvartheta = C\cdot \bfpsi$ is an eigenvector of $\Sel\Mut_{a'a'}$ with eigenvalue $\tau^{-\ell}\lambda(\bfpsi)$.
Thus
\[
\tau^\ell \lambda(\bfvartheta) = \lambda(\bfpsi),
\]
where $\lambda(\bfvartheta)$ is calculated using the mutation $\mu'(-1)=a'$, $\mu'(0)=b'$, $\mu'(1)=a'$.  
(We should use the notation $\lambda_{a'a'}(\bfvartheta)$ and $\lambda_{ac}(\bfpsi)$ to specify the mutation used, but it will be clear from context which definition we are using.)
Since $\lambda(\bfvartheta)=\rho(\bfvartheta)M(\Mut_{a'a'}\cdot \bfvartheta) < N\ell$, we have $\lambda(\bfpsi)< N\ell \tau^\ell$ as required.
\end{proof}

 Theorem \ref{main thm Asex} now follows from the following lemma.

\begin{lemma} \label{PF}
For every $a,b,c, \ell, N$, there is a unique $\bfpsi_{ac,\ell,N}\in \RR^{N^\ell}$ such that for any  non-negative, non-zero $\bfphi\in \RR^{N^\ell}$:
\[ \mbox{lim}_{t\rightarrow \infty} (\Proje\Sel\Mut_{ac})^t\cdot \bfphi = \bfpsi_{ac,\ell,N}. \] 
\end{lemma}
\begin{proof}
We apply the Perron-Frobenius theorem, which states that a non-negative, irreducible and primitive matrix has a positive (real) eigenvalue $\lambda$ whose absolute value is larger that that of any other eigenvalue, and that $\lambda$ has a unique (up to scaling) associated eigenvector  all whose coordinates are positive.
The matrix $\Sel\Mut_{ac}$ is {\em non-negative}, in the sense that all its entries are non-negative.
It is also {\em irreducible and primitive} because all the entries of $(\Sel\Mut_{ac})^N$ are positive.
So we can apply the Perron-Frobenius theorem and get a positive eigenvector $\bfpsi=\bfpsi_{ac,\ell,N}\in \RR^{N^\ell}$ which is a probability distribution with a positive eigenvalue $\lambda$ that is the largest in absolute value. As a corollary of the Perron-Frobenius theorem we also get that $\mbox{lim}_{t\rightarrow \infty} \ (\Sel\Mut_{ac})^t/\lambda^t$ is the projection to the eigenspace given by $\bfpsi$, and that this projection is non-zero for any non-zero non-negative initial population. 
This implies that $\bfpsi$ is a universal attractor of the system defined by iterating $\Sel\Mut_{ac}$ and normalisation.(For a similar application of the Perron-Frobenius Theorem in the previous literature, but which does not make use of the techniques established here to provide estimates for the mean of the resulting fixed point, see \cite{PM,PM2}.) 
\end{proof}

Before we move on to consider the asymptotic behaviour for the sex process, we need to form a stronger version of Theorem 
\ref{main thm Asex} for the 1-locus case (where the sex and asex processes are identical). While we shall not establish for general $\ell$ that $\lim_{N\to\infty}\lim_{t\to\infty}  \lambda_\asextt^t/N = \ell \tau^\ell$, we shall now do so for the case $\ell=1$ (since we shall later be able to apply this result in analysing the sex process). 

\begin{lemma} \label{added}
Let $\psi_{ac,N}$ be the probability distribution which is the fixed point of the 1-locus process, and let $\vartheta_{a'a',N} = \Proje(C \cdot \psi_{ac,N})$. Then: 
\begin{enumerate}
\item $\mbox{lim}_{N\rightarrow \infty} \lambda(\vartheta_{a'a',N})/N = 1$. 
\item $\mbox{lim}_{N\rightarrow \infty} \rho(\vartheta_{a'a',N})=1$. 
\item $\mbox{lim}_{N\rightarrow \infty} \MF(\Mut_{a'a'}\cdot \vartheta_{a'a',N})/N=1$. 
\end{enumerate}
\end{lemma}

\begin{proof} Since we consider $a,b$ and $c$ to be fixed, let $\psi_N=\psi_{ac,N}$ and $\vartheta_N=\vartheta_{a'a',N}$. 
We shall establish (1) and (2), and then (3) follows immediately from the definition of $\lambda(\vartheta_{N})$. 
The key to understanding $\vartheta_N$ is to calculate the following quotients.
For $k<N$, define:  
\[
\eta_{N}(k) =\frac{\vartheta_N(k+1)}{\vartheta_N(k)}.
\]
Let $\lambda_N=\lambda(\vartheta_N)$. Since $\vartheta_N$ is a fixed point we have that $\vartheta_N(1) = (b' \vartheta_N(1)+a' \vartheta_N(2))/\lambda_N$ and $\vartheta_N(N)= (c' \vartheta_N(N-1)+b' \vartheta_N(N))N/\lambda_N$.
It follows that: 
\[
\eta_N(1) = \frac{\lambda_N-b'}{a'}
\quad\and\quad
\eta_N({N-1}) = \frac{c'}{\lambda_N/N-b'}.
\]

\noindent For $x\notin \{ 1,N \}$ we have $\vartheta_N(x)=(c' \vartheta_N(x-1)+b' \vartheta_N(x) + a'\vartheta_N(x+1))x/\lambda_N$.
Using that $\vartheta_N(x+1)=\eta_N(x)\vartheta_N(x)$, we get: 
\[
c' \eta_N({x-1})^{-1}+ a'  \eta_N(x)=\lambda_N/x-b'.
\]

Now suppose that (1) does not hold. In this case there exists an infinite set $\Pi\subseteq \NN$, such that  $\lim_{N\in \Pi}\lambda_N/N = \kappa <1$ (note that $\lambda_N\leq N$). 
For each $x\in\NN$, define: 
\[
R(x) = \lim_{N\in \Pi} \eta_N(N-x)^{-1}.
\]
From the formulas for $\eta(N-x)$ above (and using that $a'=c'$), we deduce that $R$ satisfies the following inductive definition:
\[
R(1) = \frac{\kappa-b'}{a'}
\quad\and\quad
R(k+1) = \frac{\kappa-b'}{a'} - R(k)^{-1}.
\]
All values of $R$ are non-negative, because so are the corresponding values of $\eta_N(k)$.
Notice that $R(2)<R(1)$, and that $R(k)<R(k-1)$ implies $R(k+1)<R(k)$, from which we may conclude  that $R$ is decreasing.
$R$ must then have a limit, $\alpha$ say.
This limit must satisfy $\alpha  + \alpha^{-1}= (\kappa-b')/a'$.
Since for every $\alpha\in \RR^+$, $\alpha  + \alpha^{-1}\geq 2$, $2\leq (\kappa-b')/a'$. From the fact that  $b'=1-2a'$, it follows that $\kappa\geq 1$, which gives the required contradiction.

In order to establish (2), we show first of all that $\mbox{lim}_{N\rightarrow \infty} \vartheta_N(1)=0$. This now follows easily, however, from the fact that $\mbox{lim}_{N\rightarrow \infty} \lambda_N =\infty$ and $\eta_N(1)=(\lambda_N-b')/a'$. 

The final step is to show that $\mbox{lim}_{N\rightarrow \infty} \vartheta_N(N)=0$. Once again, consider the sequence $R(x)$ as defined above. 
We have:
\[
R(1) = \frac{1-b'}{a'}=2
\quad\and\quad
R(x+1) = 2 - R(x)^{-1}.
\]
We conclude that $R(x)> 1$ for all $x$. From this it follows that for each $x$ and all sufficiently large $N$,   $\eta_N(N-x)< 1$. This suffices to ensure that  $\mbox{lim}_{N\rightarrow \infty} \vartheta_N(N)=0$, as required. 
\end{proof}

\begin{lemma} \label{secondadded}
Let $\psi_{ac,N}$ be the probability distribution which is the fixed point of the 1-locus process. Then: 
\begin{enumerate}
\item $\mbox{lim}_{N\rightarrow \infty} \lambda(\psi_{ac,N})/N = \tau$. 
\item $\mbox{lim}_{N\rightarrow \infty} \rho(\psi_{ac,N})=1$. 
\item $\mbox{lim}_{N\rightarrow \infty} \MF(\Mut_{ac} \cdot \psi_{ac,N})/N=\tau$. 
\end{enumerate}
\end{lemma} 
\begin{proof}
Again, let $\psi_N=\psi_{ac,N}$ and let $\vartheta_N=\vartheta_{a'a',N}$ be as defined in the statement of Lemma \ref{added}. Given Lemma \ref{added}, and the fact that $\tau \lambda(\vartheta_N)=\lambda(\psi_N)$ (as established in the proof of Lemma \ref{lemma: asex fixed point bound}), it suffices to establish (2). That  $\mbox{lim}_{N\rightarrow \infty} \psi_N(N)=0$, follows from the corresponding fact for $\vartheta_N$, however, since $\psi_N \prec \vartheta_N$. It remains then, to show that $\mbox{lim}_{N\rightarrow \infty} \psi_N(1)=0$. We use a similar method to the proof of Lemma \ref{added}. This time for $k<N$, define:  
\[
\eta_{N}(k) =\frac{\psi_N(k+1)}{\psi_N(k)}.
\]
Now let $\lambda_N=\lambda(\psi_N)$. We have that: 
\[ \eta_N(1) = \frac{\lambda_N-b}{a}. \] 
The result then follows from the fact that $\mbox{lim}_{N\rightarrow \infty} \lambda_N =\infty$.  
\end{proof}

%
%
%
%
%
%


\subsubsection{The asymptotic behavior for sex}   \label{ss: limit sex}

The rest of Section \ref{ss: theorem 3} is dedicated to proving Theorem \ref{main thm Sex}, which gives  a lower bound for the geometric mean of $\lambda_\sextt^t$. For all sufficiently large $N$:
\[
 \GMinf_{t\to\infty}(\lambda_\sextt^t) > N \ell\tau^\ell.
\]
The proof requires a sequence of lemmas, some of which we will state now and  prove in later subsections.
Before describing the general architecture of the argument, let us consider how to analyse the $\ell$-locus sex process by looking at the different loci individually.
The reason we can do this is that, since the sex population stays at linkage equilibrium, its  probability distribution is determined by the product of the distributions of the individual loci.

\

\label{gvec}
Let $\bfpsi\in \RR^{N^\ell}$ be a population where all the loci are independent, as for instance after an application of recombination. 
Assume $\bfpsi$ has been normalised.
Let $ \psi_i\in \RR^N$ be the probability distribution for the $i$th locus.
We would like to analyse the $abc$-sex-process on $\bfpsi$ by analysing its process on $ \psi_i$.
It is not hard to see that if a population is at linkage equilibrium then the effect of mutation and truncation on the whole population is equivalent to considering the effect of mutation on each single locus independently as we did in the proof of Lemma \ref{lemma dagger mut}.
Let $ \psi'_i= \Mut_{ac,1}\cdot \psi_i$, where $\Mut_{ac,1}$ is the 1-locus mutation, and $\rho_i=\rho( \psi_i)=| \psi'_i|/| \psi_i|$.
If we let $\bfpsi'=\Mut_{ac}\cdot \bfpsi$, then we have that $\bfpsi'(\bfx)=\prod_{i=1}^\ell \psi'(x_i)$ and $\rho(\bfpsi)=\prod_{i=1}^\ell \rho( \psi_i)$.
Finally, we let $W_i=\MF( \psi_i')$ and $\Wha_i=\sum_{j\neq i} W_j$; let us recall that $\MF(\bfpsi')=\sum_{i=1}^\ell W_i = W_i+\Wha_i$.
From Lemma \ref{lemma: selection 1 locus} we have that the effect of selection on a single locus is given by $ \psi^*_i=\Sel_{\Wha_i}\cdot \psi_i'$, where $\Sel_{K}$ is the diagonal matrix with $\Sel_K(j,j)=(j+K)/M$.
Since this matrix actually depends also on $M$, from now on, to simplify the notation we let \[ \Sel_K(j,j)=j+K\] and leave the normalisation for later if necessary.
The increase in norm produced by $\Sel_K$ applied to  $\psi\in\RR^N$ is then $\MF(\psi)+K$. Hence  for $\psi_i'=\Mut_{ac,1}\cdot \psi_i$ we have: 
\[
\lambda(\psi_i):= \frac{|\Sel_{\Wha_i}\cdot \Mut_{ac,1} \cdot \psi_i |}{|\psi_i |} = \rho(\psi_i)(\MF(\psi'_i)+\Wha_i). 
\]

\

The next step is to describe how we are going to use the translation to the $a'=c'$ case for sex populations. 
Consider the following setting.
Let $\{K^t\}_{t\in\NN}$ be a sequence of real numbers in $[0,(\ell-1)N]$ (which will later represent the sequence $\Wha^t_i$ for some fixed $i$).  
Let $\psi^0\in \RR^N$ (later this will represent the initial sexual population at some locus), and let $\vartheta^0=C\cdot \psi^0$.
Assume that $\psi^0$ is not the zero vector. For every $t\in\NN$, define: 
\[
\psi^{t+1}=\Sel_{K^t}\cdot \Mut_{ac,1}\cdot \psi^t
\quad\mbox{ and }\quad
\vartheta^{t+1}=\Sel_{K^t}\cdot \Mut_{a'a',1}\cdot \vartheta^t.
\]
From Lemma \ref{lemma: conjugating mutation} we have $\tau^{t} \cdot \vartheta^t =  C\cdot \psi^t$ for every $t$.
Define:
\[
\lambda^t_{\psi}= \frac{|\psi^{t+1}|}{|\psi^t|} = \rho(\psi^t)\cdot (\MF((\psi^t)^{\prime})+K^t)
\quad\mbox{ and }\quad
\lambda^t_{\vartheta}=\frac{|\vartheta^{t+1}|}{|\vartheta^t|} = \rho(\vartheta^t)\cdot (\MF((\vartheta^t)^{\prime})+K^t),
\]
where $(\psi^t)^{\prime}=\Mut_{ac,1}\cdot \psi^t$ and $(\vartheta^t)^{\prime}=\Mut_{a'a',1}\cdot \vartheta^t$.

If  $\psi^t$ was a fixed point, then using that $\tau^{t} \cdot \vartheta^t =  C\cdot \psi^t$ we could conclude that $\vartheta^t$ is also a fixed point (all given the appropriate normalisations), and that $\lambda^t_\psi=\tau \lambda^t_{\vartheta}$ as in the proof of Lemma \ref{lemma: asex fixed point bound}.
Even without assuming that the process converges to a limit distribution, we still get that these values have the same geometric means:

\begin{lemma}\label{lemma: GMs} $\GM_{t\to\infty}(\lambda^t_{\psi}/\lambda^t_{\vartheta})=\tau.$
\end{lemma}
\begin{proof}
For every $k$, we have that:
\[
 \frac{\tau^{k} |\vartheta^k|}{|\vartheta^0|} = \frac{|C\cdot \psi^k|}{| C\cdot \psi^0|} = \frac{| \psi^k|}{|  \psi^0|}\frac{|C\cdot \Proje(\psi^k)|}{| C\cdot \Proje(\psi^0) |}.
\]
The set $\{\phi\in ({\RR^{\geq 0}})^N: |\phi|=1\}$ is compact and hence the image of the continuous map $\psi\mapsto |C\cdot \Proje(\psi)|$ is a closed interval of the form $[\alpha, \alpha\beta]$ for $0<\alpha$ and $1\leq\beta$ (we get that $\alpha>0$ because $|C\cdot \Proje(\psi))|$ is always positive).
We then have:
\[
\GM_{t< k}(\lambda^t_{\psi}) = \sqrt[k]{\frac{| \psi^k|}{|  \psi^0|} }\leq   \sqrt[k]{\beta \frac{\tau^{k} | \vartheta^k|}{|  \vartheta^0|}}  =  \tau \sqrt[k]{\beta}\ \GM_{t< k}(\lambda^t_{\vartheta}) . 
\]
Symmetrically $\GM_{t< k}(\lambda^t_{\psi}) \geq \tau\sqrt[k]{\beta^{-1}} \GM_{t< k}(\lambda^t_{\vartheta})$.
The lemma then follows from the fact that both $\sqrt[k]{\beta^{-1}}$ and $\sqrt[k]{\beta}$ converge to $1$ as $k\to\infty$.
\end{proof}

The next step is to give an approximate calculation for $\lambda^t_{\vartheta}$, which holds irrespective of the choice for $\psi^0$.
The next lemma shows that  $\lambda^t_{\vartheta}$ is eventually always close to $N+K^t$.
The reason this holds is that $\rho(\vartheta^t)$ is eventually always close to 1, and  $\MF((\vartheta^t)^{\prime})$ is eventually always close to $N$ because positive mutations are as likely as negative ones.  

\begin{lemma}\label{lemma: boundaries}
For any choice of non-negative $\psi^0\in \RR^N $ and $\{K^t\}_{t\in\NN}$, let $\{\lambda_{\vartheta}^t\}_{t\in \NN}$ be defined as above.  There is a sequence $\{\epsilon_N\}_{N\in\NN}$ converging to 0 such that, for every $N$, every sequence $\{K^t\}_{t\in\NN}$ and every non-negative $\psi^0$, the following holds for all sufficiently large $t$: 
\[
1-\epsilon_N < \frac{\lambda^t_{\vartheta}}{N+K^t}  < 1. 
\]
\end{lemma}

The proof of this lemma is a little technical, so we delay it until Subsection \ref{sss: boundaries}.
The following is a small lemma concerning geometric means, which will allow us to compare the geometric means of $\lambda^t_{\vartheta}$ and $K^t$.

\begin{lemma}\label{lemma GM of sums}
Let $\{a^t\}_{t<k}$ be a sequence of positive real numbers and let $b$ be a positive number.
Then $b + \GM_{t< k}(a^t)  \leq  \GM_{t<k}(b+a^t)$.
\end{lemma}
\begin{proof}
This is a corollary of Jensen's inequality which states that $\varphi(k^{-1}\sum_{i=1}^k\gamma_i)  \leq k^{-1}\sum_{i=1}^k\varphi(\gamma_i)$ for $\varphi\colon\RR\to\RR$ convex.
One has to apply it to the values $\gamma_i=\log(a^i)$ and the function $\varphi(x)=\log(b+e^{x})$ which is convex because $\varphi''(x)=b e^x/(b+e^x)^2 > 0$.
\end{proof}

For a given choice of $\psi^0$, we can use what we have so far to get a lower bound for $\GM(\rho(\psi^t))$. For all sufficiently large $k$:
\begin{eqnarray}
\GM_{t<k}\rho(\psi^t) &=& \GM_{t<k}\frac{\lambda^t_{\psi}}{M((\psi^t)^{\prime})+K^t} \\ 
					&\geq& \GM_{t<k}\frac{\tau\lambda^t_{\vartheta}}{M((\psi^t)^{\prime})+K^t} (1-\epsilon_N)\\
					&\geq& \tau\left(\GM_{t<k}\frac{N+K^t}{M((\psi^t)^{\prime})+K^t}\right)(1-\epsilon_N)^2.    \label{eq: GMrho}
\end{eqnarray}
The value inside the large parentheses is greater than 1 for large $N$ and sufficiently large $t$ because there exists a sequence $\{\epsilon_N \}_{N\in \NN}$ with limit 0 such that, for large $t$, $M((\psi^t)^{\prime})<\tau N+\epsilon_N<N$, as proved in the next lemma.
The lemma also shows that, for large $N$ and for sufficiently large $t$, the value of $\psi^t$ at the upper boundary $N$ is very close to $0$.

\begin{lemma}\label{lemma: M below tau N} 
For any choice of $\psi^0\in \RR^N-\{ \boldsymbol{0} \}$ and $\{K^t\}_{t\in\NN}$, let $\{\psi^t\}_{t\in \NN}$ be defined as above, and let $(\psi^t)^{\prime}=\Mut_{ac,1}\cdot \psi^t$. 
There exists a sequence $\{\epsilon_N\}_{N\in\NN}$ converging to 0 such that, for every $N$, every sequence $\{K^t\}_{t\in\NN}$ and every $\psi^0$, the following holds for all sufficiently large $t$: \[  M((\psi^t)^{\prime})<\tau N +\epsilon_N\ \ \mbox{and} \ \ \psi^t(N)/|\psi^t|<\epsilon_N.\]  
\end{lemma}
\begin{proof}
Let $\phi^0\in \RR^N$ be the probability distribution with  $\phi^0(N)=1$. For each $t$ define  $\phi^{t+1}=\Sel_0\Mut_{ac,1}\cdot \phi^t$.
By induction on $t$ and using Lemmas \ref{mutpres} and \ref{lemma: Sel preserves prec} we conclude that $\psi^t\preceq \phi^t$ and $(\psi^t)^{\prime} \preceq (\phi^t)^{\prime}$ for all $t$. Using Lemma \ref{lemma prec E} we then get that $M((\psi^t)^{\prime})\leq M((\phi^t)^{\prime})$.  Applying Lemma \ref{secondadded}, we conclude that there exists a sequence $\{\epsilon_N\}_{N\in \NN}$ with limit 0 such that  $M((\phi^t)^{\prime})$ (and so also $M((\psi^t)^{\prime})$) remains below $\tau N+\epsilon_N$ for all sufficiently large $t$. The second claim of the lemma also follows from Lemma \ref{secondadded} and the fact that $\psi^t \preceq \phi^t$. 
\end{proof}

For the last stretch of the proof we need to be more concrete about the sexual population we are analysing.
Let $\bfpsi^0\in \RR^{N^\ell}$ be the initial sexual population, and $\{\bfpsi^t\}_{t\in\NN}$ be the sequence obtained  by iterating $ac$-mutation, selection and recombination.
For each locus $i$ and generation $t$, let $\psi^t_i$ be the distribution at locus $i$ at stage $t$, but ignoring normalisation.
We use the same notation we have been using so far:
\begin{itemize}
\item $(\psi^t_i)^{\prime}=\Mut_{ac,1}\cdot \psi^t_i$;
\item $(\bfpsi^t)^{\prime}=\Mut_{ac}\cdot \bfpsi^t$;
\item $W^t_i=\MF((\psi_i^t)^{\prime})$;
\item $M^t=\MF((\bfpsi^t)^{\prime})=\sum_{i=1}^\ell W^t_i$;
\item $\Wha^t_i=M^t_i-W^t_i$; 
\item $\psi^{t+1}_i = \Sel_{\Wha^t_i}\cdot \Mut_{ac,1} \cdot \psi^t_i$;
\item $\rho(\psi^t_i) = |(\psi_i^t)^{\prime}|/|\psi^{t}_i|$, 
\item $\rho(\bfpsi^t) = |(\bfpsi^t)^{\prime}|/|\bfpsi^{t}| = \prod_{i=1}^\ell \rho(\psi^t_i)$.
\end{itemize}

The objective now is to show that the geometric mean of $\lambda(\bfpsi^t) = \rho(\bfpsi^t) M^t = (\prod_{i=1}^\ell \rho(\psi^t_i)) M^t$ is above $N\ell\tau^\ell$.
We will apply the results we have obtained thus far for $K^t=\Wha^t_i$. In order to be able to do this, however, we need to be able to compare $\Wha^t_i$ and $M^t$.
If all loci were identical, we would have $\Wha^t_i=M^t (\ell-1)/\ell$.
When the loci are not identical, the following lemma gives us an approximation to $N+\Wha^t_i$, and tells us that it is close -- at least in geometric mean -- to $N+M^t(\ell-1)/\ell$, just as it would be if the loci were identical.

\begin{lemma}\label{lemma: Whas} 
There is a sequence $\{\epsilon_N:N\in\NN\}$ converging to 0 such that for every $N$ and every initial population $\bfpsi^0\in \RR^{N^\ell}$ as above, the following holds for all sufficiently large $k$: 
\[
1-\epsilon_N < \GM_{t<k}\left( \frac{N+ \Wha_i^t}{N+ M^t ((\ell-1)/\ell)}\right) <  1+\epsilon_N. 
\]
\end{lemma}
\noindent We will prove this lemma in Subsection \ref{sss: Whas}. Lemmas \ref{lemma: boundaries},   \ref{lemma: M below tau N} and \ref{lemma: Whas}   all assert the existence of certain sequences $\{\epsilon_N\}_{N\in \NN}$ with limit 0. We now let $\{ \epsilon_N \}_{N\in \NN}$ be a sequence with limit 0, which majorises each of the sequences provided by these lemmas. 

We are now ready to finish the  proof of Theorem \ref{main thm Sex}.
Let $k$ be large.
We start by cleaning up equation (\ref{eq: GMrho}) using what we now know from Lemmas \ref{lemma GM of sums} and  \ref{lemma: Whas}.
Fix $i\leq \ell$.
\begin{eqnarray*}
\GM_{t<k} \rho^t_i    &\geq&    \tau\left(\GM_{t\leq k}\frac{N+\Wha_i^t}{W_i^t+\Wha_i^t}\right)(1-\epsilon_N)^2 \\
				& \geq  &   \tau\left(\GM_{t\leq k}\frac{N+M^t(\ell-1)/\ell}{M^t}\right)(1-\epsilon_N)^3 \\
				& = & \tau\left(\GM_{t\leq k}\left(\frac{\ell-1}{\ell}+ \frac{N}{M^t}   \right) \right)(1-\epsilon_N)^3\\
				& \geq & \tau\left(\frac{\ell-1}{\ell}+ \GM_{t\leq k}\left(\frac{N}{M^t}   \right) \right)(1-\epsilon_N)^3\\
				&=& \tau\left(1+\xi_{k}/\ell \right)(1-\epsilon_N)^3,
\end{eqnarray*}
where $\xi_{k}= \GM_{t\leq k}(\ell N/M^t) -1$.
Notice that $\GM_{t\leq k}(\ell N/M^t) > 1$. Furthermore, by Lemma \ref{lemma: M below tau N}, $M^t< \ell(\tau N+\epsilon_N)$ for sufficiently large $t$. Adjusting the sequence $\epsilon_N$ as necessary (but maintaining the fact that it has limit 0), we then have that for  sufficiently large  $k$,  $\xi_{k}\geq ((1-\tau)/\tau) -\epsilon_N$.

Finally,
\begin{eqnarray}
\GM_{t < k}(\lambda_\sextt^t)     &=&    \GM_{t < k}( M^t)  \prod_{i=1}^\ell \GM_{t < k}(\rho^t_i)  \\
							&\geq &   \frac{N\ell}{1+\xi_{k}} \left(\tau\left(1+\xi_{k}/\ell \right)(1-\epsilon_N)^3\right)^\ell      \\
							&=& \left(N\ell\tau^\ell\right) \left(\frac{(1+\xi_{k}/\ell)^\ell}{1+\xi_{k}}\right)  (1-\epsilon_N)^{3\ell},
\end{eqnarray}
for all large enough $k$. The last observation to make is that there exists $\epsilon>0$, independent of $N$ and $k$, such that the factor 
\[
\left(\frac{(1+\xi_{k}/\ell)^\ell}{1+\xi_{k}}\right)  (1-\epsilon_N)^{3\ell}
\]
is greater than $1+\epsilon$ for large  $N$ and sufficiently large $k$. To see this, note that
the function $x\mapsto (1+x/\ell)^\ell/(1+x)$ is always greater than $1$ for $x>0$ and tends to $+\infty$ as $x\to+\infty$.
(It is actually increasing for $x>0$.) Let $N^{\ast}$ be large enough that $\epsilon_{N} <(1-\tau)/\tau$ for all $N>N^{\ast}$. 
Among all the $x$'s with $x\geq ((1-\tau)/\tau)-\epsilon_{N^{\ast}}$, there is a minimum possible value for $(1+x/\ell)^\ell/(1+x)$, call it $\zeta$, which is greater than 1. Let $\epsilon$ be such that $\zeta =1+2\epsilon$. 
Then for $N\geq N^{\ast}$  for which $\epsilon_N$ is sufficiently small, we have $\left(\frac{(1+\xi_{k}/\ell)^\ell}{1+\xi_{k}}\right)  (1-\epsilon_N)^{3\ell}>1+\epsilon$ for all large enough $k$.

\

It remains to prove Lemmas \ref{lemma: boundaries} and \ref{lemma: Whas}.

\subsubsection{The proof of Lemma \ref{lemma: boundaries}}  \label{sss: boundaries}
Roughly speaking, we need to show that $\lambda^t_{\vartheta}$ gets close to $N+K^t$ as $t$ becomes large.
Recall that $\lambda^t_{\vartheta} = \rho(\vartheta^t)(\MF((\vartheta^t)^{\prime})+K^t)$, and that $\rho(\vartheta^t)=1-a 
\vartheta^t(1)/|\vartheta^t| - c \vartheta^t(N)/|\vartheta^t|$.
The proof will have three parts which are: showing that $\MF((\vartheta^t)^{\prime})$ gets close to $N$, showing that $\vartheta^t(1)/|\vartheta^t|$ gets close to 0, and showing that $\vartheta^t(N)/|\vartheta^t|$ gets close to 0.
We remark that it is not surprising that $\MF(\vartheta^t)$ gets close to $N$: since positive and negative mutations are equiprobable, mutation without truncation does not bring the mean fitness down, while selection only ever increases mean fitness. Therefore the mean fitness can be expected to rise, this rise being halted only by effect of truncation at the upper boundary.

The first idea for the proof is to consider an alternative population which  evolves according to a different sequence $\{K^t\}_{t\in\NN}$; one that is constant, and  that is either always larger or else always smaller than the original one.
The fact that the sequence is constant allows us to apply the Perron-Frobenius theorem and establish a limit population, which we can later analyse.
Choosing a sequence $K^t$ with larger (smaller) values will guarantee that the new sequence is $\prec$-below (-above) the original. This  allows us to compare the mean fitnesses of the two populations, as well as their values at the boundaries 1 and $N$.

Let us begin with the analysis of the limit populations.
Fix a value of $K$, for which we will later substitute either 0 or $(\ell-1)N$. Define
\[
\phi^0= \vartheta^0
\quad\and\quad
\phi^{t+1}=\Sel_{K}\Mut_{a'a', 1} \cdot \phi^t.
\]
Since $\phi^t$ is defined by iterating a linear system which is non-negative, primitive and irreducible, we can apply the Perron-Frobenius theorem, exactly as we did in Lemma \ref{PF}, to deduce that the populations $\phi^t$ must converge to a limit population $\phi_N$ which is independent of the starting population (and depends only on $N$, $K$ and $a'$).
In order to analyse $\phi_N$, we proceed much as in the proof of Lemma \ref{added}. Once again, the key idea is to consider quotients between consecutive values in the distribution. For $k<N$, define: 
\[
\eta_{N}(k) =\frac{\phi_N(k+1)}{\phi_N(k)}.
\]
Let $\lambda_N=\lambda(\phi_N)=\rho(\phi_N)(\MF((\phi_N)^{\prime})+K)$.
Since  $\phi_N(1) = (b' \phi_N(1)+a' \phi_N(2))(1+K)/\lambda_N$ and  $\phi_N(N)= (c' \phi_N(N-1)+b' \phi_N(N))(N+K)/\lambda_N$ we have:
\[
\eta_N(1) = \frac{\lambda_N/(1+K)-b'}{a'}
\quad\and\quad
\eta_N({N-1}) = \frac{c'}{\lambda_N/(N+K)-b'}.
\]
For $x\notin \{1,N \}$  we have  $\phi_N(x)=(c' \phi_N(x-1)+b' \phi_N(x) + a'\phi_N(x+1))(x+K)/\lambda_N$.
Since  $\phi_N(x+1)=\eta_N(x)\phi_N(x)$ this gives:  
\[
c' \eta_N({x-1})^{-1}+ a'  \eta_N(x)=\lambda_N/(x+K)-b'.
\]

Let us now move into the proof that the mean fitness grows close to $N$.
Consider $K=(\ell-1)N$, and define $\phi_N$ as above for that $K$.
Since $K^t\leq (\ell-1)N$ (where $K^t$ is the sequence given in the statement of the lemma), it follows by induction using Lemmas \ref{mutpres} and \ref{lemma: Sel preserves prec} that for every $t$, $\phi^t \prec \vartheta^t$. This means that  $\phi^t(1)/|\phi^t|\geq \vartheta^t(1)/|\vartheta^t|$, and (by Lemma \ref{lemma prec E}) that $\MF(\phi^t)\leq \MF(\vartheta^t)$.
In order to establish that $\lim_N\lambda_N/(\ell N) = 1$, suppose otherwise. Then there must exist an infinite set $\Pi$, such that $\mbox{lim}_{N\in \Pi} \lambda_N/(\ell N) = \kappa <1$ (note that $\lambda_N\leq \ell N$). 
For each $x\in\NN$, define: 
\[
R(x) = \lim_{N\in \Pi} \eta_N(N-x)^{-1}.
\]
From the formulas for $\eta(N-x)$ above (and using that $a'=c'$), it follows that each $R(x)$ is defined and satisfies the following inductive definition:
\[
R(1) = \frac{\kappa-b'}{a'}
\quad\and\quad
R(k+1) = \frac{\kappa-b'}{a'} - R(k)^{-1}.
\]
All the values of $R$ are non-negative, because so are the corresponding values of $\eta_N(k)$.
Note that $R(2)<R(1)$, and that $R(k)<R(k-1)$ implies $R(k+1)<R(k)$, from which we conclude that $R$ is decreasing.
$R$ must then have a limit, $\alpha$ say.
This limit must satisfy $\alpha  + \alpha^{-1}= (\kappa-b')/a'$.
Since for every $\alpha\in \RR^+$, $\alpha  + \alpha^{-1}\geq 2$, we have that $2\leq (\kappa-b')/a'$.  Since $b'=1-2a'$, it follows that $\kappa\geq 1$, which gives the required contradiction.

So far we have concluded that  $\mbox{lim}_{N\rightarrow \infty} \lambda_N/N = \ell$. 
Since $\lambda_N=\rho(\phi_N)(\MF(\phi_N^{\prime})+(\ell-1)N)$ and $\rho(\phi_N)\leq 1$, it follows that $\mbox{lim}_{N\rightarrow \infty} \MF(\phi_N')/N = 1$.
For now, let $\epsilon_N=1 - \MF(\phi_N')/N$.
Since $(\phi^t)^{\prime} \preceq (\vartheta^t)^{\prime}$, we also know that   $\liminf_t\MF((\vartheta^t)^{\prime})/N\geq 1-\epsilon_N$.

The second step is to show that $\vartheta^t(1)/|\vartheta^t|$ is small for large $t$.
Since $\phi^t\preceq \vartheta^t$, we know that $\phi^t(1)/|\phi^t| \geq \vartheta^t(1)/|\vartheta^t|$, so it is enough to show that once normalised $\phi_N(1)$ is small for large $N$. 
This time we define: 
\[
R(k) = \lim_{N\to\infty} \eta_N(1).
\]
We have that: 
\[
R(1) = \frac{\ell/(\ell-1)-b'}{a'}
\quad\and\quad
R(k+1) = \frac{\ell/(\ell-1)-b'}{a'} - R(k)^{-1}.
\]
Since $(\ell/(\ell-1)-b')/a'>(1-b')/a'=2$ it follows inductively that $R(k)> 1$ for all $k$.
This means that for every $k$, there exists $N$ large enough such  that $\eta_{N'}(x)> 1$ for all $N'\geq N$ and all  $x\leq k$.
Redefine $\epsilon_N$ to be the maximum between the value $\epsilon_N$ specified in the above and $1/k$ for the largest $k$ such that $\eta_N(x)> 1$ for all $x\leq k$.
It follows that for that for all $N$, $\phi_N(1)\leq \epsilon_N$ and that the sequence $\epsilon_N$ converges to 0.

The third step is to consider $\vartheta^t(N)$.
This time we set $K=0$ and consider the new corresponding sequence $\phi^t$, with the new limit $\phi_N$.
We now have that $\vartheta^t\preceq \phi^t$, and hence that $\vartheta^t(N)/|\vartheta^t| \leq \phi^t(N)/|\phi^t|$.
This time, for each $x$ we define:
\[
R(x) = \lim_{N\in \Pi} \eta_N(N-x)^{-1}.
\]
By the same argument as above we get that $\mbox{lim}_{N\rightarrow \infty} \lambda_N/(N+K) = 1$, and in this case this means that $\mbox{lim}_{N\rightarrow \infty} \lambda_N/N = 1$.
$R$ now satisfies:
\[
R(1) = \frac{1-b'}{a'}=2
\quad\and\quad
R(x+1) = 2 - R(x)^{-1}.
\]
Again we have that $R(k)> 1$ for all $k$, which means that for sufficiently large  $N$, $\eta_N(N-x)< 1$ for all $x\leq k$.
We can therefore redefine $\epsilon_N$ so that this sequence still converges to 0 and:  
\[
\liminf_t \frac{\rho(\vartheta^t)(\MF((\vartheta^t)^{\prime})+K^t)} {N+K^t}\geq 1-\epsilon_N,
\]
as needed for Lemma \ref{lemma: boundaries}.

\subsubsection{The proof of Lemma \ref{lemma: Whas}}  \label{sss: Whas}
In this section we prove the last lemma required to complete the proof of Theorem \ref{main thm Sex}.  
Roughly speaking,  Lemma \ref{lemma: Whas} asserts that the various $W_i$'s (for varying $i$) eventually stay relatively close to each other, even if they are initially quite different.
In simulations we have observed that in fact all of the  $W_i$'s converge to the same value $M/\ell$ (see for instance Figure 5), but this seems to be hard to prove.
Instead, we prove that $N+\Wha_i^t$ becomes close to $N+ M^t(\ell-1)/\ell$ in geometric mean, which is enough for our purposes.

Let us begin by looking at a 2-locus $ac$-sex process where selection acts with an additive value $K^t$ at stage $t$.
More formally, let $\{K^t:t\in\om\}$ be a sequence of numbers in $[0, (\ell-2)N]$, let $\upsilon^0_0,\upsilon^0_1\in \RR^N$ be the initial distributions corresponding to each of those two loci, and define: 
\[
\upsilon^{t+1}_0= \Sel_{M((\upsilon^t_1)^{\prime})+K^t}\Mut_{ac,1}\cdot \upsilon^t_0
\quad\text{and}\quad
\upsilon^{t+1}_1= \Sel_{M((\upsilon^t_0)^{\prime})+K^t}\Mut_{ac,1}\cdot \upsilon^t_1.
\]
Notice how the the $M$-value used in defining selection at a given locus is the mean corresponding to the other locus (as it should be for the 2-locus sex process).

\begin{lemma}\label{lemma: delta t}
There is a sequence $\{\delta^t: t\in\om\}$ such that each $\delta^t>1$,  with $\GMsup_{t\to\infty}\delta^t<1+\epsilon_N$ and such that for all $t$: 
\[
\frac{1}{\delta^t} \leq \frac{N+M((\upsilon^t_1)^{\prime})+K^t}{N+M((\upsilon^t_0)^{\prime})+K^t}\leq \delta^t. 
\]
\end{lemma}
\begin{proof}
Let $\phi^0_0, \phi^0_1\in \RR^N$ be new initial populations, such that  $\phi^0_0$ is the probability distribution with $\phi^0_0(1)=1$ and $\phi^0_1$ is the probability distribution with $\phi^0_1(N)=1$. 
Let the $\phi^t_0$ and $\phi^t_1$ processes evolve as follows: 
\[
\phi^{t+1}_0 = \Sel_{M((\phi^t_1)^{\prime})+K^t}\Mut_{ac,1} \cdot \phi^t_0
\quad\and\quad
\phi^{t+1}_1 = \Sel_{M((\phi^t_0)^{\prime})+K^t}\Mut_{ac,1} \cdot \phi^t_1.
\]
Consider the translations of $\phi_0$ and $\phi_1$ to the $a'a'$-process: i.e., let $\vartheta^t_0=\tau^{-t} C\cdot \phi^t_0$ and $\vartheta^t_1=\tau^{-t} C\cdot \phi^t_1$.
From  Lemmas \ref{lemma: GMs} and \ref{lemma: boundaries} we get that:
\[
1-\epsilon_N < \frac{\GM(\rho(\phi^t_1)(M((\phi^t_1)^{\prime}) + \MF((\phi^t_0)^{\prime})+K^t)}{\tau \GM(N+ \MF((\phi^t_0)^{\prime})+K^t)} < 1+\epsilon_N,
\]
and 
\[
1-\epsilon_N < \frac{\GM(\rho(\phi^t_0)(\MF((\phi^t_0)^{\prime}) + \MF((\phi^t_1)^{\prime})+K^t)}{\tau \GM(N+ \MF((\phi^t_1)^{\prime})+K^t)} < 1+\epsilon_N.
\]
Taking the quotient of these equations we conclude that we can redefine the sequence $\epsilon_N$ so that it still converges to 0, and so that: 
\[
(1-\epsilon_N) < \GM\left(\frac{\rho(\phi^t_1)}{\rho(\phi^t_0)}\right) \GM\left(\frac{N+ \MF((\phi^t_1)^{\prime})+K^t}{N+ \MF((\phi^t_0)^{\prime})+K^t}\right) < (1+\epsilon_N).
\]
From  Lemmas \ref{mutpres} and \ref{lemma: Sel preserves prec}, it follows inductively that $\phi^t_0\preceq \upsilon^t_0\preceq \phi^t_1$ and $\phi^t_0\preceq \upsilon^t_1\preceq \phi^t_1$ for every $t$.
We will now use the fact that $\phi^t_0\preceq \phi_1^t$ to establish that the numerators above are essentially greater than the denominators. 
We know from Lemma \ref{lemma prec E} that $(\phi^t_0)^{\prime}\preceq (\phi_1^t)^{\prime}$ implies $\MF((\phi^t_1)^{\prime})\geq \MF((\phi^t_0)^{\prime})$.
Also,   $\phi^t_0\preceq \phi_1^t$ implies that $\phi^t_1(1)/|\phi^t_1| \leq  \phi^t_0(1)/|\phi^t_0| $. We know that $\phi^t_1(N)/|\phi^t_1|<\epsilon_N$ from Lemma  \ref{lemma: M below tau N} (with $\epsilon_N$ as specified there).
We therefore have that:   
\[
\rho(\phi^t_1)= 1 - a \phi^t_1(1)/|\phi^t_1| - c \phi^t_1(N)/|\phi^t_N| \geq 1 - a \phi^t_0(1)/|\phi^t_0| - c \phi^t_0(N)/|\phi^t_N| -c\epsilon_N = \rho(\phi^t_0)-c\epsilon_N.
\] 
Since $\rho(\phi^t_0)>1-a-c=b>c$, it follows that $\rho(\phi^t_0)-c\epsilon_N\geq \rho(\phi^t_0)(1-\epsilon_N)$.
We can therefore redefine the $\epsilon_N$ so that the sequence still converges to 0 and: 
\[
1\leq \GM\left(\frac{N+ \MF((\phi^t_1)^{\prime})+K^t}{N+ \MF((\phi^t_0)^{\prime})+K^t}\right) < (1+\epsilon_N).
\]
Let $\delta^t$ be the term inside the large parentheses, i.e., $\delta^t=\frac{N+ \MF((\phi^t_1)^{\prime})+K^t}{N+ \MF((\phi^t_0)^{\prime})+K^t}\geq 1$ and $\GM(\delta_t)< 1+\epsilon_N$.
Since $\MF((\phi^t_0)^{\prime})\leq \MF((\upsilon^t_0)^{\prime}) \leq \MF((\phi^t_1)^{\prime})$ and $\MF((\phi^t_0)^{\prime})\leq \MF((\upsilon^t_1)^{\prime}) \leq \MF((\phi^t_1)^{\prime})$, we have: 
\[
\frac{1}{\delta^t} \leq \frac{N+M((\upsilon^t_1)^{\prime})+K^t}{N+M((\upsilon^t_0)^{\prime})+K^t}\leq \delta^t.    
\qedhere\]
\end{proof}

Now let us return to the proof of Lemma \ref{lemma: Whas}.
For each $i<j$, let $\delta_{i,j}^t$ be the $\delta^t$ whose existence is ensured by Lemma \ref{lemma: delta t} for the case  $K^t=\Wha^t_{i,j} = M^t-W_i^t-W_j^t$.
Let $\delta^t=\prod_{i<j\leq\ell}\delta^t_{i,j}$.
We have that $\GM(\delta^t)<(1+\epsilon_N)^{\ell(\ell-1)}$ and that for every $t$ and $i\neq j$: 
\[
\frac{N+\Wha^t_i}{N+\Wha^t_j}<\delta^t.
\] 
Since $\sum_{i=1}^\ell(N+\Wha^t_i)= \ell N + (\ell-1) M^t$, it follows that for some $j$, $N+\Wha^t_j>N + M^t(\ell-1)/\ell$.  Therefore,  for every $i$, $N+\Wha^t_i>(N + M^t(\ell-1)/\ell)/\delta^t$.
A similar argument shows that $N+\Wha^t_i<(N + M^t(\ell-1)/\ell)\delta^t$, which completes the proof of Lemma \ref{lemma: Whas}.



\subsection{Variants of the model} \label{variants}

In this subsection we briefly consider variants of the model for which populations may be finite or infinite and fitnesses may be additive or multiplicative.  For the most part, our analysis here will rely on the results of simulations, although we shall also be able to draw some concrete conclusions concerning fundamental similarities and differences between the models. First, let us describe these variants.

\subsubsection{The finite model}
For the finite model we consider an extra parameter $P\in \NN$, which determines the size of the population. This size is then fixed through the generations, so that a population always consists of $P$ vectors, $\bfx_1,....,\bfx_P$, in $\ZZ^\ell$.  Let us consider the sex process first. 
In order to apply selection one chooses $2P$ individuals, by sampling independently from the population with replacement: if $M$ is the mean fitness of the population and $F(\bfx)$ is the fitness of individual $\bfx$, then the probability that individual $\bfx$ is chosen for the $n$th sample ($1\leq n \leq 2P$) is $F(\bfx)/M$.  One may consider the parent generation as forming a pool of gametes.  The probability that a gamete chosen uniformly at random from this pool comes from a given individual $\bfx$, is proportional to the fitness of $\bfx$. During the selection phase we are choosing $P$ many pairs of individuals from which gametes are taken (recombination later being applied to each of these pairs).  To apply the mutation operation, we take in turn each individual from the $P$-many pairs chosen during selection, and for each locus we change its fitness value by $-1$, $0$ or $1$ with probabilities $\mu(-1)$, $\mu(0)$ and $\mu(1)$ respectively.
To apply recombination, we take each of the $P$ pairs resulting from mutation in turn. Suppose that the $n$th pair is $\bfx_{n}^1=(x_{n,1}^1,...,x_{n,\ell}^1)$ and  $\bfx_{n}^2=(x_{n,1}^2,...,x_{n,\ell}^2)$. Then we form $\bfx_n^*$ which is the $n$th member of the next generation  by taking each locus $i$ in turn and defining either $x_{n,i}^*=x_{n,i}^1$ or $x_{n,i}^*=x_{n,i}^2$, each with equal probability.  The assumption of maximum recombination rates might be justified by considering that one is choosing a representative gene from each chromosome, meaning that the monitored genes lie on distinct chromosomes.   For the asex process, one proceeds similarly, except that $P$ many individuals rather than pairs are chosen during the selection phase, and the recombination phase is omitted. 

The finite model is clearly the most important to understand, and the analysis we have provided for the infinite model provides a good approximation for large populations and over a number of generations which is not too large. As mentioned previously, 
the equations governing the change in mean fitness and variance due to selection and mutation for the infinite population model would now perfectly describe the \emph{expected} effect of mutation and selection for finite populations, and the finite populations model could be seen simply as a stochastic approximation to the infinite case, were it not for the loss in variance and higher cumulants due to sampling. While the effect of sampling may not be too significant for large populations on a stage by stage basis, long term it will have the effect that the mean fitness no longer tends to infinity over stages. Without providing a rigorous proof, one may reason that this is perhaps unsurprising as follows. Mutation will still have a fixed expected effect on the mean and variance at each stage. For a population $\bfphi$ with $M=M(\bfphi)$, $V=V(\bfphi)$ and $\kappa_3=\kappa_3(\bfphi)$, however, while the expected effect of selection on the mean is just as for the infinite population model, the expected effect of selection on variance is now: 
\[ \left( \frac{\kappa_3}{M}- \left( \frac{V}{M} \right)^2\right)\frac{P-1}{P} - \frac{V}{P}. \] 

\noindent Now the ratios between cumulants will not tend to increase without limit (in the infinite populations model simulations show these ratios converging to fixed values over time, and such behaviour is also approximated for the finite model). Thus, if variance was to increase without limit, selection would soon produce decreases in variance outweighing any increases given by mutation. A similar analysis can be made including the effect of recombination, establishing that for sufficiently large variances, the effect of sampling will outweigh any other increases in variance.

\subsubsection{The multiplicative model}
This model is defined exactly like the additive one with the sole difference that the fitness of an individual is calculated multiplicatively, i.e., $F(X)=\prod_{i=1}^\ell X_i$. For the infinite case, the sex and asex processes now behave identically, given populations initially at linkage equilibrium. This was initially observed by Maynard-Smith \cite{MS2}. 

\begin{lemma}
Multiplicative selection preserves linkage equilibrium.  
\end{lemma}
\begin{proof}
Suppose that $\bfphi$ is a population at linkage equilibrium. Let $X_1,...,X_\ell$ be random variables distributed according to $\bfphi$, and let $X_1^{\ast},...,X_{\ell}^{\ast}$ be distributed according to $\bfphi^{\ast}=\Sel(\bfphi)$. We show that selection maintains independence between the first two loci, as the general result is  very similar. We must show that whenever $\mbox{\textbf{P}}(X_1^{\ast}=m_1)\neq 0$ and $\mbox{\textbf{P}}(X_1^{\ast}=m_2)\neq 0$: 

\[ \mbox{\textbf{P}}(X_2^{\ast}=n|\ X_1^{\ast}=m_1)=  \mbox{\textbf{P}}(X_2^{\ast}=n|\ X_1^{\ast}=m_2). \] 

\noindent This is equivalent to: 

\[ \frac{\mbox{\textbf{P}}(X_1^{\ast}=m_1\ \wedge X_2^{\ast}=n)}{\mbox{\textbf{P}}(X_1^{\ast}=m_1)} =  \frac{\mbox{\textbf{P}}(X_1^{\ast}=m_2\ \wedge X_2^{\ast}=n)}{\mbox{\textbf{P}}(X_1^{\ast}=m_2)}. \] 

\noindent Now, since selection acts according to multiplicative fitnesses: 
\[  \frac{\mbox{\textbf{P}}(X_1^{\ast}=m_1)}{\mbox{\textbf{P}}(X_1^{\ast}=m_2)}= \frac{m_1}{m_2}  \frac{\mbox{\textbf{P}}(X_1=m_1)}{\mbox{\textbf{P}}(X_1=m_2)}. \] 
Also: 

\[  \frac{\mbox{\textbf{P}}(X_1^{\ast}=m_1\ \wedge X_2^{\ast}=n)}{\mbox{\textbf{P}}(X_1^{\ast}=m_2\ \wedge X_2^{\ast}=n)}= \frac{nm_1}{nm_2}  \frac{\mbox{\textbf{P}}(X_1=m_1\ \wedge X_2=n)}{\mbox{\textbf{P}}(X_1=m_2\ \wedge X_2=n)}, \] 
\noindent so the result follows from linkage equilibrium for $\bfphi$. 
\end{proof}

Thus, if a population begins at linkage equilibrium, this linkage equilibrium will be preserved throughout all the stages (for the $\NN$- and bounded-models). Each application of recombination now has no effect on the population. 

For the finite multiplicative model, however, sampling will produce linkage disequilibrium and sex now robustly outperforms asex (as seen in the outcomes of simulations presented in \S\ref{EXTDAT}). For an insightful analysis of mechanisms which may allow negative $LD_2$ to build up in this context see \cite{BO}.





\subsection{The simulations} \label{ss: simulations}
 For a small number of loci $\ell$ one can implement the algorithms described directly. If one wishes to deal with a larger number of loci for the infinite population sex process then one can achieve  more efficient simulations (for the $\NN$ and bounded models, and which will give only tiny margins of error due to truncation issues for the $\ZZ$ model), by making use of Lemma \ref{lemma: selection 1 locus}, which allows one to track the entire population by monitoring each locus separately. 
Similarly, one can achieve more efficient simulations for the asex infinite population process (for the $\ZZ$-model, and with only tiny margins of error for the other domains) by monitoring only the distribution on the total fitness of individuals. For finite populations, such mechanisms for improving efficiency are not generally necessary (or indeed possible). 
In considering the unbounded infinite populations models, of course one can only deal with a bounded domain in practice. One is therefore limited in the number of generations which can be simulated.  
To make the computations more precise for the infinite bounded model, we represented real numbers by their logarithms, as the values of the probability distribution at the upper and lower bounds are extremely small.\\

\noindent \textbf{Author contributions}. Both authors contributed equally to the construction of proofs and simulations. \\

\noindent \textbf{Author Information}. Correspondence and request for materials should be addressed either to andy@aemlewis.co.uk or antonio@math.berkeley.edu.


\begin{thebibliography}{1}


\bibitem{Kond} Kondrashov, A.S.\ Classification of hypotheses on the advantage of amphimixis. \emph{J.\ Hered.} 84: 372-387 (1993).  

\bibitem{May1} Maynard-Smith, J. \emph{The evolution of sex}. Cambridge University Press, Cambridge (1978). 

\bibitem{Fel} Felsenstein, J. The evolutionary advantage of recombination. \emph{Genetics} 78: 737-756 (1974). 

\bibitem{Bern} Bernstein, C.\ \& Bernstein, H. \emph{Aging, sex and DNA repair}. San Diego Academic Press (1991). 

\bibitem{Mich} Michod, R.E. Genetic error, sex and diploidy. \emph{J.\ Hered.} 84: 360-371.  

\bibitem{May2} Maynard-Smith, J. The evolution of recombination. In \emph{The evolution of sex: an examination of current ideas}, Michod, R.E.\ \& Levin, B.R.\ eds. Sunderland, Massachusetts, Sinauer: 106-125 (1988). 

\bibitem{BC} Barton, N.H.\ \& Charlesworth, B. Why sex and recombination? \emph{Science} 281: 1986--1990 (1998). 













\bibitem{OB} Otto, S. P., \& Barton, N.H.  The evolution of recombination: removing the limits to natural selection. \emph{Genetics} 147:879-906 (1997).



\bibitem{BO} Barton, N.H., \& Otto, S.P.  Evolution of recombination due to random drift. \emph{Genetics} 169:2353-2370 (2005).

\bibitem{HR} Hill, W.G., \& Robertson, A.\  The effect of linkage on the
limits to artificial selection. \emph{Genetical Research} 8:269-294 (1966).

\bibitem{Morg} Morgan, T. H.  \emph{Heredity and sex}. Columbia University Press, New York (1913).

\bibitem{Fish} Fisher, R. A. \emph{The genetical theory of natural selection}. Oxford
University Press, Oxford (1930).


\bibitem{Mull} Muller, H. J.\  Some genetic aspects of sex. \emph{American Naturalist} 66:118-138 (1932).

\bibitem{Mull2} Muller, H.J.\ The relation of recombination to mutational advance. \emph{Mutation Research} 1:2-9 (1964).

\bibitem{SKB} Salath\'{e}, M., Kouyos, R.D.\ \& Bonhoeffer, S. On the causes of selection for recombination underlying the Red Queen hypothesis. \emph{American Naturalist} 174(suppl.):S31-S42 (2009).

\bibitem{PL} Peters, A. D., \& Lively, C.M. The Red Queen and fluctuating epistasis: a population genetic analysis of antagonistic coevolution. \emph{American Naturalist} 154:393-405 (1999).

\bibitem{GO} Gandon, S., \& Otto, S.P. The evolution of sex and recombination in response to abiotic or coevolutionary fluctuations in epistasis. \emph{Genetics} 175:1835-1853 (2007).

\bibitem{PZF} Pylkov, K. V., Zhivotovsky, L.A.\  \& Feldman, M.W.  Migration
versus mutation in the evolution of recombination under multi-locus selection. \emph{Genetical Research} 71:247-256 (1998).

\bibitem{LO} Lenormand, T., \& Otto, S.P.  The evolution of recombination in a heterogeneous environment. \emph{Genetics} 156:423-438 (2000).

\bibitem{Otten} Otto, S.P.\  The evolutionary enigma of sex. \emph{American Naturalist}, 174: S1-S14 (2009).

\bibitem{OK} Ohta,T., \& Kimura,M. A model of mutation appropriate to estimate the number of electrophoretically detectable alleles in a finite population. \emph{Genet.\ Res.}\ 22, 201-204, (1973). 

\bibitem{PM} Moran, P.A.P. Global stability of genetic systems governed by mutation and selection, \emph{Math. Proc. Camb. Phil. Soc.}, 80, 331-336 (1976). 

\bibitem{PM2} Moran, P.A.P. Global stability of genetic systems governed by mutation and selection II, \emph{Math. Proc. Camb. Phil. Soc.}, 81, 435-441 (1977). 

\bibitem{GY} Geodakyan, V.A., The role of sexes in the transmission and transformation of genetic information, \emph{Probl.\ Peredachi Inf.}, 1, 105-112, (1965) (In Russian). 

\bibitem{DL} Lloyd, D.L., Benefits and handicaps of sexual reproduction, \emph{Evol Biol}, 13, 69-111, (1980). 

\bibitem{LP} Livnat, A., Papadimitriou, C., Dushoff, J., \& Feldman, M.W., A mixability theory for the role of sex in evolution, \emph{PNAS}, vol.\  105, no.\ 50, 19803-19808, (2008). 

\bibitem{LP2} Livnat, A., Papadimitriou, C., Pippenger, N., \& Feldman, M.W., Sex, mixability and modularity, \emph{PNAS}, vol.\ 107, no.\ 4, 1452-1457, (2009). 

\bibitem{PB} Pr\"{u}gel-Bennett, A., When a genetic algorithm outperforms hill-climbing,  \emph{Theoretical Computer Science}, 320, (1), 135-153, (2004).

\bibitem{MS2} Maynard-Smith, Evolution in sexual and asexual populations, \emph{American Naturalist}, 102, 469-678, (1968). 

\bibitem{CK} Crow, J.F., \& Kimura, M. The theory of genetic loads, \emph{Proc XI Int.\  Congr.\ Genetics.}, vol 2, 495-505, (1964).   

\end{thebibliography}
\end{document}